\documentclass{article}%
\usepackage{amsmath}
\usepackage{amsfonts}
\usepackage{amssymb}
\usepackage{graphicx}
\usepackage{float}
\usepackage{tikz}
\usepackage{tkz-graph}
\usepackage{algorithm}
\usepackage{algpseudocode}
\usepackage{placeins}%
\providecommand{\U}[1]{\protect\rule{.1in}{.1in}}
\numberwithin{equation}{section}
\newtheorem{theorem}{Theorem}[section]

\newtheorem{definition}[theorem]{Definition}
\newtheorem{example}[theorem]{Example}

\newtheorem{lemma}[theorem]{Lemma}

\newenvironment{proof}[1][Proof]{\noindent \textbf{#1.} }{\  \rule{0.5em}{0.5em}}
\begin{document}

\title{Generalized Cops and Robbers:\\A Multi-Player Pursuit Game on Graphs}
\author{Ath. Kehagias\thanks{The author thanks Steve Alpern and Pascal Schweitzer for
several useful and inspiring discussions.}}
\date{\today}
\maketitle

\begin{abstract}
We introduce and study the \emph{Generalized Cops and Robbers} game (GCR), an
$N$-player pursuit game in graphs. The two-player version is essentially
equivalent to the classic \emph{Cops and Robbers} (CR)\ game. The three-player
version can be understood as two CR games played simultaneously on the same
graph; a player can be at the same time both \emph{pursuer and evader}. The
same is true for four or more players. We formulate GCR as a \emph{discounted
stochastic game of perfect information } and prove that, for three or more
players, it has at least two \emph{Nash Equilibria}:\ one in positional
deterministic strategies and another in non-positional ones. We also study the
capturing properties of GCR Nash Equilibria\ in connection to the
\emph{cop-number} of a graph. Finally, we briefly discuss GCR as a member of a
wider family of multi-player graph pursuit games with rather interesting properties.

\end{abstract}

\section{Introduction\label{sec01}}

We introduce and study \emph{Generalized Cops and Robbers} (\emph{GCR}); it
is\ a \emph{multi-player} \emph{pursuit game} closely related to the classic
two-player \emph{Cops and Robbers} (\emph{CR})\ game
\cite{nowakowski1983,quillot1979}.

GCR is played on a finite, simple, undirected graph $G$ by $N$ players
$P_{1},P_{2},...,P_{N}$ (with $N\geq2$). The players start at given vertices
of the graph and at each turn one player moves to a vertex in the closed
neighborhood of his current position; the other players stay at their current
vertices. The game effectively terminates when, for some $n\in\left\{
1,2,...,N-1\right\}  $,\ $P_{n}$ \emph{captures} $P_{n+1}$, i.e., when they
are located in the same vertex; if no capture ever takes place, the game
continues \emph{ad infinitum}.

Let us denote the GCR game with $N$ players by $\Gamma_{N}$. Then $\Gamma_{2}$
is very similar to the classic CR\ game, where $P_{1}$ (the \textquotedblleft
cop\textquotedblright)\ tries to capture $P_{2}$ (the \textquotedblleft
robber\textquotedblright). In $\Gamma_{3}$, $P_{1}$ tries to capture $P_{2}$
who tries to evade $P_{1}$ and capture $P_{3}$; $P_{1}$ can never be captured
and $P_{3}$ can never capture. Hence $\Gamma_{3}$ can be understood as two CR
games played simultaneously on the same graph; a player is both pursuer and
evader \emph{at the same time}. The situation is extended similarly for higher
$N$ values.

As we will show, $\Gamma_{2}$ can be formulated as a \emph{zero-sum stochastic
game} which has a \emph{value} (and both players have \emph{optimal
strategies}). On the other hand, for $N\geq3$, $\Gamma_{N}$ is a non-zero sum
game and the main question is the existence of Nash Equilibria (NE). As we
will show, more than one such equilibria always exist and they sometimes lead
to surprising player behavior. In this sense, GCR presents novel and (we
hope)\ mathematically interesting problems.

There is a rich literature on pursuit games in graphs, Euclidean spaces and
other more general structures but it is generally confined to two-player games.

The seminal works on pursuit games in graphs are
\cite{nowakowski1983,quillot1979} in which the classic CR game was introduced.
A great number of variations of the classic game have been studied; an
extensive and recent review of the related literature appears in the book
\cite{bonato2011}. However, practically all of this literature concerns
\emph{two-player} games. Classic CR and its variants may involve more than one
cops, but all of them are \emph{tokens} controlled by a single \emph{cop
player}. A very interesting paper \cite{bonato2017} deals with
\textquotedblleft generalized cops and robber games\textquotedblright\ but
again the scope is restricted to two-player games. In fact, the only previous
work (of which we are aware)\ dealing with \emph{multi-player} games of
pursuit in graphs is our own \cite{kehagias2017}. It is also remarkable that,
while classic CR\ and many of its variants admit a natural game theoretic
formulation and study, this has not been exploited in the CR\ literature.

Regarding pursuit in Euclidean spaces, the predominant approach is in terms of
\emph{differential games} as introduced in the seminal book \cite{isaacs1965}.
There is a flourishing literature on the subject, which contains many works
involving multiple pursuers, but they are generally assumed to be
\emph{collaborating} \cite{chen2016,jang2005,pham2008,souidi2015,sun2017}. The
case of antagonistic pursuers has been studied in some papers
\cite{foley1974,ho1969} but the methods used in these works do not appear to
be easily applicable to the study of pursuit / evasion on \emph{graphs}.

This paper is organized as follows. Section \ref{sec02} is preliminary:\ we
introduce notation, define \emph{states}, \emph{histories} and
\emph{strategies} and give a general form of the \emph{payoff} function. In
Section \ref{sec03} we prove that, for any graph and any number of players,
GCR\ has a NE in \emph{deterministic positional strategies}; this result is
applicable not only to GCR\ but to a wider family of pursuit games, as will be
discussed later. In Section \ref{sec04} we show that in the two-player GCR
game: (i)\ the \emph{value} of the game exists (essentially it is the
logarithm of the \emph{optimal capture time)\ }and (ii)\ both players have
optimal deterministic positional strategies. Because of the close connection
of GCR\ to the classical CR\ game, these results also hold for CR; while they
have been previously established by graph theoretic methods, we believe our
proof is the first game-theoretic one. In Section \ref{sec05} we study the
three-player GCR game and prove: (i)\ the existence of a NE in deterministic
positional strategies; (ii)\ the existence of an additional \ NE in
deterministic but non-positional strategies; (iii) various results connecting
the classic \emph{cop number} of a graph to \emph{capturability}. In Section
\ref{sec06} we briefly discuss $N$-players GCR when $N\geq4$. In Section
\ref{sec07} we show that the ideas behind GCR\ can be generalized to obtain a
large family of multi-player pursuit games on graphs. Finally, in Section
\ref{sec08} we summarize, present our conclusions and discuss future research directions.

\section{Preliminaries\label{sec02}}

The following notations will be used throughout the paper.

\begin{enumerate}
\item Given a graph $G=\left(  V,E\right)  $, for any $x\in V$, $N\left(
x\right)  $ is the \emph{neighborhood} of $x$: $N\left(  x\right)  =\left\{
y:\left\{  x,y\right\}  \in E\right\}  $; $N\left[  x\right]  $ is the
\emph{closed neighborhood} of $x$: $N\left[  x\right]  =N\left(  x\right)
\cup\left\{  x\right\}  $.

\item The cardinality of set $A$ is denoted by $\left\vert A\right\vert $; the
set of elements of $A$ which are not elements of $B$ is denoted by
$A\backslash B$.

\item $\mathbb{N}$ is the set of natural numbers $\left\{  1,2,3,...\right\}
$ and $\mathbb{N}_{0}$ is $\left\{  0,1,2,3,...\right\}  $. For any
$M\in\mathbb{N}$ we define $\left[  M\right]  =\left\{  1,2,...,M\right\}  $.

\item The \emph{graph distance} (length of shortest path in $G$)\ between
$x,y\in V$ is denoted by $d_{G}\left(  x,y\right)  $ or simply by $d\left(
x,y\right)  $.
\end{enumerate}

In Section \ref{sec01} we have described GCR\ informally; now we define the
elements of the game rigorously.

The game proceeds at discrete \emph{turns} (time steps) and at every turn all
players except one must remain at their locations. In other words, at every
turn $t\in\mathbb{N}$, for every player except one, the \emph{action set} (see
(\ref{eq007}) below)\ is a singleton. This, in addition to the fact that all
players are aware of all previously executed moves, means that GCR\ is a
\emph{perfect information game}.

Any player $P_{n}$ can have the first move, but afterwards they move in the
sequence implied by their numbering:
\[
P_{n}\rightarrow P_{n+1}\rightarrow...\rightarrow P_{N}\rightarrow
P_{1}\rightarrow P_{2}\rightarrow....
\]
When a player has the move, he can either move to a vertex adjacent to his
current one or stay in place. Hence the \emph{game position} or \emph{game
state} has the form $s=\left(  x^{1},x^{2},...,x^{N},p\right)  $ where
$x^{n}\in V$ is the position (vertex) of the $n$-th player and $p\in\left[
N\right]  $ is the number of the player who has the next move. The set of
\emph{nonterminal} states is%
\[
S=\left\{  \left(  x^{1},x^{2},...,x^{N},p\right)  :\left(  x^{1}%
,x^{2},...,x^{N}\right)  \in V\times V\times...\times V\text{ and }p\in\left[
N\right]  \right\}  .
\]
We introduce an additional \emph{terminal state} $\overline{s}$. Hence the
full state set is
\[
\overline{S}=S\cup\left\{  \overline{s}\right\}  .
\]
We define $S_{n}$ to be the set of states in which $P_{n}$ has the next move:%
\[
\text{for each }n\in\left[  N\right]  :S^{n}=\left\{  s:s=\left(  x^{1}%
,x^{2},...,x^{N},n\right)  \in S\right\}  ,
\]
Hence the set of nonterminal states can be partitioned as follows:%
\[
S=S^{1}\cup S^{2}\cup...\cup S^{N}.
\]
For any $n\in\left[  N-1\right]  $, we say that $P_{n}$ \emph{captures}
$P_{n+1}$ iff they are located in the same vertex; the set of $P_{n}%
$\emph{-capture states}, i.e., those in which $P_{n}$ captures $P_{n+1}$ is
$\widetilde{S}^{n}$:
\[
\text{for each }n\in\left[  N-1\right]  :\widetilde{S}^{n}=\left\{
s:s=(x^{1},x^{2},...,x^{N},p)\in S\text{ and }x^{n}=x^{n+1}\right\}  .
\]
Hence nonterminal states can be partitioned into two sets:%
\[%
\begin{tabular}
[c]{ll}%
$\text{capture states: }$ & $S_{C}=\widetilde{S}^{_{1}}\cup\widetilde{S}%
^{_{2}}\cup...\cup\widetilde{S}^{N-1}\text{ ,}$\\
non-capture states: & $S_{NC}=S\backslash S_{C}.$%
\end{tabular}
\ \ \
\]

As already mentioned, when $P_{n}$ has the move, he can move to any vertex in
the closed neighborhood of $x^{n}$; when another player has the move, $P_{n}$
can only stay in place; when the game is in a capture state or in the terminal
state, every player has only the \textquotedblleft null move\textquotedblright%
\ $\lambda$. Formally, when the game state is $s$, the $n$-th player's
\emph{action set} is denoted by $A^{n}\left(  s\right)  $ and defined by
\begin{equation}
A^{n}\left(  s\right)  =\left\{
\begin{array}
[c]{ll}%
N\left[  x^{n}\right]  & \text{when }s=\left(  x^{1},x^{2},...,x^{N},n\right)
\in S^{n}\cap S_{NC},\\
\left\{  x^{n}\right\}  & \text{when }s=\left(  x^{1},x^{2},...,x^{N}%
,m\right)  \in S^{m}\cap S_{NC}\text{ with }m\neq n,\\
\left\{  \lambda\right\}  & \text{when }s\in S_{C}\cup\left\{  \overline
{s}\right\}  .
\end{array}
\right.  \label{eq007}%
\end{equation}
The players' \emph{actions} (i.e., \emph{moves}) effect state-to-state
transitions in the obvious manner. Suppose the game is at position $s\in
S^{n}$ and $P_{n}$ makes the move $a^{n}\in A_{n}\left(  s\right)  $; then
$\mathbf{T}\left(  s,a^{n}\right)  $ denotes the resulting game position. A
capture state always transits to $\overline{s}$ and $\overline{s}$ always
transits to itself:%
\[
\forall s\in S_{C}:\mathbf{T}\left(  s,\lambda\right)  =\overline{s}%
\qquad\text{ and}\qquad\mathbf{T}\left(  \overline{s},\lambda\right)
=\overline{s}.
\]
We define the \emph{capture time} to be
\[
T_{C}=\min\left\{  t:x_{t}^{1}=x_{t}^{2}\text{ or }x_{t}^{2}=x_{t}^{3}\text{
or ... or }x_{t}^{N-1}=x_{t}^{N}\right\}  .
\]
If no capture takes place, the \emph{capture time} is $T_{C}=\infty$. Hence
the game can evolve as follows.

\begin{enumerate}
\item \noindent If $T_{C}=0$ then the initial state $s_{0}$ is a capture state
and $s_{t}=\overline{s}$ for every $t\in\mathbb{N}=\left\{  1,2,...\right\}  $.

\item If $0<T_{C}<\infty$ then:

\begin{enumerate}
\item at the $0$-th turn the game starts at some preassigned state $s_{0}\in
S_{NC}$;

\item at the $t$-th turn (for $0<t<T_{C}$), the game moves to some state
$s_{t}\in S_{NC}$;

\item at the $T_{C}$-th turn the game moves to some capture state $s_{T_{C}%
}\in S_{C}$ and

\item at $t=T_{C}+1$ the game moves to the terminal state and stays there: for
every $t>T_{C}$, $s_{t}=\overline{s}$.
\end{enumerate}

\item \noindent Finally, i\noindent f $T_{C}=\infty$ then $s_{t}\in S_{NC}$
for every $t\in\mathbb{N}_{0}=\left\{  0,1,2,...\right\}  $.
\end{enumerate}

\noindent According to the above, the game starts at some preassigned state
$s_{0}=\left(  x_{0}^{1},x_{0}^{2},x_{0}^{3},p_{0}\right)  $ and at the $t$-th
turn ($t\in\mathbb{N}$) is in the state $s_{t}=\left(  x_{t}^{1},x_{t}%
^{2},...,x_{t}^{N},p_{t}\right)  $. This results in a \emph{game history}
$\mathbf{s=}s_{0}s_{1}s_{2}...$ . In other words, we assume each play of the
game lasts an infinite number of turns; however, if $T_{C}<$ $\infty$ then
$s_{t}=\overline{s}$ for every $t>T_{C}$; hence, while the game lasts an
infinite number of turns, it \emph{effectively} ends at $T_{C}$. We define the
following history sets.

\begin{enumerate}
\item Histories of length $k:H_{k}=\left\{  \mathbf{s}=s_{0}s_{1}%
...s_{k}\right\}  $;

\item Histories of finite length$:H_{\ast}=\cup_{k=1}^{\infty}H_{k}$;

\item Histories of infinite length$:H_{\infty}=\left\{  \mathbf{s}=s_{0}%
s_{1}...s_{k}...\right\}  $.
\end{enumerate}

A \emph{deterministic strategy} is a function $\sigma^{n}$ which assigns a
move to each finite-length history:
\[
\sigma^{n}:H_{\ast}\rightarrow V
\]
At the start of the game $P_{n}$ selects a $\sigma^{n}$\emph{ }which
determines all his subsequent moves. We will only consider \emph{legal}%
\footnote{I.e., they never produce moves outside the player's action set.}
deterministic strategies\footnote{As will be seen, since GCR is a game of
perfect information, the player loses nothing by using only deterministic
strategies.}. A \emph{strategy profile} is a tuple $\sigma=\left(  \sigma
^{1},\sigma^{2},...,\sigma^{N}\right)  $, which specifies one strategy for
each player. We are particularly interested in \emph{positional }strategies,
i.e., $\sigma^{n}$ such that the next move depends only on the current state
of the game (but not on previous states or current time):%
\[
\sigma^{n}\left(  s_{0}s_{1}...s_{t}\right)  =\sigma^{n}\left(  s_{t}\right)
.
\]
We define $\sigma^{-n}=\left(  \sigma^{j}\right)  _{j\in\left[  N\right]
\backslash\left\{  n\right\}  }$; for instance, if $\sigma=\left(  \sigma
^{1},\sigma^{2},\sigma^{3}\right)  $ then $\sigma^{-1}=\left(  \sigma
^{2},\sigma^{3}\right)  $.

To complete the description of GCR we must specify the players' \emph{payoff
functions}; we will do this in several steps. In this section we give a
general form of the payoff function, which applies not only to GCR,\ but to a
broader family of $N$-player games (with $N\geq2$). In the next section we
will prove that any game of this family admits at least one NE in
\emph{positional deterministic }strategies. In subsequent sections we will
treat separately the cases of GCR$\ $with $N=2$, $N=3$ and $N\geq4$ players;
in each case, by completely specifying the payoff function, we will reach
additional conclusions regarding the properties of the respective game.

For the time being we only specify that the \emph{total payoff} function of
the $n$-th player ($n\in\left[  N\right]  $)\ has the form%
\begin{equation}
Q^{n}\left(  s_{0},\sigma\right)  =\sum_{t=0}^{\infty}\gamma^{t}q^{n}\left(
s_{t}\right)  , \label{eq02001}%
\end{equation}
where: $q^{n}$ is the \emph{turn payoff }(it depends on $s_{t}$, the game
state at time $t$) which is assumed to be bounded:%
\[
\exists M:\forall n\in\left[  N\right]  ,\forall s\in S:\left\vert
q^{n}\left(  s\right)  \right\vert \leq M;
\]
and $\gamma\in\left(  0,1\right)  $ is the \emph{discount factor}.

Since the total payoff is the sum of the discounted turn payoffs, GCR\ is a
\emph{multi-player discounted stochastic game }\cite{filar1996}. Recall that a
stochastic game is one which consists of a sequence of one-shot games, each of
which depends on the previous game played and the actions of the players. In
GCR the players can limit themselves to deterministic strategies; since the
state transitions are also deterministic, while GCR is a \textquotedblleft
stochastic game\textquotedblright\ in the above sense, in all cases of
interest it will actually evolve in a deterministic manner.

We will denote by $\Gamma_{N}\left(  G|s_{0}\right)  $ the GCR\ game played by
$N$ players on graph $G$, starting from state $s_{0}$. Our results hold for
any $\gamma\in\left(  0,1\right)  $ so, for simplicity of notation, we omit
the $\gamma$ dependence. In addition, $\gamma$ will be omitted from statements
of theorem, lemmas etc. in the rest if the paper, since$\ $all the results
presented hold for any $\gamma\in\left(  0,1\right)  $.

\section{Nash Equilibria for Perfect Information Discounted
Games\ \label{sec03}}

The following theorem shows that every $\Gamma_{N}\left(  G|s_{0}\right)  $
has a Nash Equilibrium in deterministic positional strategies.

\begin{theorem}
\label{prop0301}For every graph $G$, every $N\geq2$ and every initial state
$s_{0}\in S$ the game $\Gamma_{N}\left(  G|s_{0}\right)  $ admits a profile of
deterministic positional strategies $\widehat{\sigma}=\left(  \widehat{\sigma
}^{1},\widehat{\sigma}^{2},...,\widehat{\sigma}^{N}\right)  $ such that%
\begin{equation}
\forall n\in\left[  N\right]  ,\forall s_{0}\in S,\forall\sigma^{n}%
:Q^{n}\left(  s_{0},\widehat{\sigma}^{n},\widehat{\sigma}^{-n}\right)  \geq
Q^{n}\left(  s_{0},\sigma^{n},\widehat{\sigma}^{-n}\right)  . \label{eq02011}%
\end{equation}
For every $s$ and $n$, let $u^{n}\left(  s\right)  =Q^{n}\left(
s,\widehat{\sigma}\right)  $. Then the \ following equations are satisfied%
\begin{align}
\forall n,\forall s  &  \in S^{n}:\widehat{\sigma}^{n}\left(  s\right)
=\arg\max_{a^{n}\in A^{n}\left(  s\right)  }\left[  q^{n}\left(  s\right)
+\gamma u^{n}\left(  \mathbf{T}\left(  s,a^{n}\right)  \right)  \right]
,\label{eq02012}\\
\forall n,m,\forall s  &  \in S^{n}:u^{m}\left(  s\right)  =q^{m}\left(
s\right)  +\gamma u^{m}\left(  \mathbf{T}\left(  s,\widehat{\sigma}^{n}\left(
s\right)  \right)  \right)  . \label{eq02013}%
\end{align}

\end{theorem}

\begin{proof}
Fink has proved in \cite{fink1963} that every $N$-player discounted stochastic
game has a positional NE in \emph{probabilistic} strategies; this result holds
for the general game (i.e., with \emph{concurrent} moves and probabilistic
strategies and state transitions). According to \cite{fink1963}, at
equilibrium the following equations must be satisfied for all $m$ and $s$:
\begin{equation}
\mathfrak{u}^{m}\left(  s\right)  =\max_{\mathbf{p}^{m}\left(  s\right)  }%
\sum_{a^{1}\in A^{1}\left(  s\right)  }\sum_{a^{2}\in A^{2}\left(  s\right)
}...\sum_{a^{N}\in A^{N}\left(  s\right)  }p_{a^{1}}^{1}\left(  s\right)
p_{a^{2}}^{2}\left(  s\right)  ...p_{a^{N}}^{N}\left(  s\right)  \left[
q^{m}\left(  s\right)  +\gamma\sum_{s^{\prime}}\Pi\left(  s^{\prime}%
|s,a^{1},a^{2},...,a^{N}\right)  \mathfrak{u}^{m}\left(  s^{\prime}\right)
\right]  , \label{eq02015}%
\end{equation}
where we have modified Fink's original notation to fit our own; in particular:

\begin{enumerate}
\item $\mathfrak{u}^{m}\left(  s\right)  $ is the expected value of
$u^{m}\left(  s\right)  $;

\item $p_{a^{m}}^{m}\left(  s\right)  $ is the probability that, given the
current game state is $s$, the $m$-th player plays action $a^{m}$;

\item $\mathbf{p}^{m}\left(  s\right)  =\left(  p_{a^{m}}^{m}\left(  s\right)
\right)  _{a^{m}\in A^{m}\left(  s\right)  }$ is the vector of all such
probabilities (one probability per available action);

\item $\Pi\left(  s^{\prime}|s,a^{1},a^{2},...,a^{N}\right)  $ is the
probability that, given the current state is $s$ and the player actions are
$a^{1},a^{2},...,a^{N}$, the next state is $s^{\prime}$ .
\end{enumerate}

\noindent Now choose any $n$ and any $s\in S^{n}$. For all $m\neq n$, the
$m$-th player has a single move, i.e., we have $A^{m}\left(  s\right)
=\left\{  a^{m}\right\}  $, and so $p_{a_{m}}^{m}\left(  s\right)  =1$. Also,
since transitions are deterministic,%
\[
\sum_{s^{\prime}}\Pi\left(  s^{\prime}|s,a^{1},a^{2},...,a^{N}\right)
\mathfrak{u}^{n}\left(  s^{\prime}\right)  =\mathfrak{u}^{n}\left(
\mathbf{T}\left(  s,a^{n}\right)  \right)  .
\]
Hence, for $m=n$, (\ref{eq02015})\ becomes
\begin{equation}
\mathfrak{u}^{n}\left(  s\right)  =\max_{\mathbf{p}^{n}\left(  s\right)  }%
\sum_{a^{n}\in A^{n}\left(  s\right)  }p_{a^{n}}^{n}\left(  s\right)  \left[
q^{n}\left(  s\right)  +\gamma\mathfrak{u}^{n}\left(  \mathbf{T}\left(
s,a^{n}\right)  \right)  \right]  . \label{eq02016}%
\end{equation}
Furthermore let us define $\widehat{\sigma}^{n}\left(  s\right)  $ (for the
specific $s$ and $n$) by
\begin{equation}
\widehat{\sigma}^{n}\left(  s\right)  =\arg\max_{a^{n}\in A^{n}\left(
s\right)  }\left[  q^{n}\left(  s\right)  +\gamma\mathfrak{u}^{n}\left(
\mathbf{T}\left(  s,a^{n}\right)  \right)  \right]  . \label{eq02017}%
\end{equation}
If (\ref{eq02016}) is satisfied by more than one $a^{n}$, we set
$\widehat{\sigma}^{n}\left(  s\right)  $ to one of these arbitrarily. Then, to
maximize the sum in (\ref{eq02016}) the $n$-th player can set $p_{\widehat
{\sigma}^{n}\left(  s\right)  }^{n}\left(  s\right)  =1$ and $p_{a}^{n}\left(
s\right)  =0$ for all $a\neq\widehat{\sigma}^{n}\left(  s\right)  $. Since
this is true for all states and all players (i.e., every player can, without
loss, use deterministic strategies) we also have $\mathfrak{u}^{n}\left(
s\right)  =u^{n}\left(  s\right)  $. Hence (\ref{eq02016}) becomes
\begin{equation}
u^{n}\left(  s\right)  =\max_{a^{n}\in A^{n}\left(  s\right)  }\left[
q^{n}\left(  s\right)  +\gamma u^{n}\left(  \mathbf{T}\left(  s,a^{n}\right)
\right)  \right]  =q^{n}\left(  s\right)  +\gamma u^{n}\left(  \mathbf{T}%
\left(  s,\widehat{\sigma}^{n}\left(  s\right)  \right)  \right)  .
\label{eq02018}%
\end{equation}
For $m\neq n$, the $m$-th player has no choice of action and (\ref{eq02016})
becomes
\begin{equation}
u^{m}\left(  s\right)  =q^{m}\left(  s\right)  +\gamma u^{m}\left(
\mathbf{T}\left(  s,\widehat{\sigma}^{n}\left(  s\right)  \right)  \right)  .
\label{eq02019}%
\end{equation}
We recognize that (\ref{eq02017})-(\ref{eq02019}) are (\ref{eq02012}%
)-(\ref{eq02013}). Also, (\ref{eq02017}) defines $\widehat{\sigma}^{n}\left(
s\right)  $ for every $n$ and $s$ and so we have obtained the required
deterministic positional strategies $\widehat{\sigma}=\left(  \widehat{\sigma
}^{1},\widehat{\sigma}^{2},\widehat{\sigma}^{3}\right)  $.
\end{proof}

Note that the initial state $s_{0}$ plays no special role in the system
(\ref{eq02012})-(\ref{eq02013}). In other words, using the notation $u\left(
s\right)  =\left(  u^{1}\left(  s\right)  ,u^{2}\left(  s\right)
,...,u^{N}\left(  s\right)  \right)  $ and $\mathbf{u}=\left(  u\left(
s\right)  \right)  _{s\in S}$, we see that $\mathbf{u}$ and $\widehat{\sigma}$
are the same for every starting position $s_{0}$ and every game $\Gamma
_{N}\left(  G|s_{0}\right)  $ (when $N,G$ and $\gamma$ are fixed).

Fink's proof requires that, for every $n$, the total payoff is $Q^{n}\left(
s_{0},\sigma\right)  =\sum_{t=0}^{\infty}\gamma^{t}q^{n}\left(  s_{t}\right)
$; but does not place any restrictions (except boundedness)\ on $q^{n}$. The
same is true of our proof; hence Theorem \ref{prop0301} applies not only to
the GCR\ game, for which the form of $q^{n}$\ will be specified in Sections
\ref{sec04}, \ref{sec05} and \ref{sec06}, but to a wider family of games,
which will be discussed in Section \ref{sec07}.

\section{GCR\ with Two Players and (Classic CR)\label{sec04}}

We now proceed to a more detailed study of $\Gamma_{2}\left(  G|s_{0}\right)
$. To this end, we first specify the form of the turn payoff functions $q^{1}$
and $q^{2}$:%
\begin{equation}
q^{1}\left(  s\right)  =-q^{2}\left(  s\right)  =\left\{
\begin{array}
[c]{rl}%
1 & \text{iff }s\in\widetilde{S}^{1}\\
0 & \text{else.}%
\end{array}
\right.  .
\end{equation}
Recalling that $T_{C}$ is the capture time (and letting $\gamma^{\infty}=0$),
for every $s_{0}$ and deterministic $\sigma$ which result in capture at time
$T_{C}$, we clearly have:%
\[
Q^{1}\left(  s_{0},\sigma\right)  =-Q^{2}\left(  s_{0},\sigma\right)
=\gamma^{T_{C}}.
\]
So $\Gamma_{2}\left(  G|s_{0}\right)  $ is a zero-sum game. Furthermore, since
$\log Q^{1}\left(  s_{0},\sigma\right)  =T_{C}\log\gamma$ and $\log\gamma<0$,
it follows that $P_{1}$ (resp. $P_{2}$)\ will maximize his payoff by
minimizing (resp. maximizing)\ capture time $T_{c}$. Hence we have the
following simple description:\ 

\begin{quotation}
\noindent$\Gamma_{2}\left(  G|s_{0}\right)  $ is a two-player game in which,
starting from an initial position $s_{0}=\left(  x^{1},x^{2},p\right)  $,
$P_{1}$ attempts to capture $P_{2}$ in the shortest possible time and $P_{2}$
attempts to delay capture as long as possible.
\end{quotation}

\noindent This is true whenever both players use deterministic strategies,
which they can do without loss since $\Gamma_{2}\left(  G|s_{0}\right)  $ is a
perfect information \ game. In particular, according to Theorem \ref{prop0301}%
, this holds when they play optimally. In fact, according to Theorem
\ref{prop0301} (for every $G$ and $s_{0}$) $\Gamma_{2}\left(  G|s_{0}\right)
$ has a NE $\widehat{\sigma}=\left(  \widehat{\sigma}^{1},\widehat{\sigma}%
^{2}\right)  $ in deterministic positional strategies. And, since $\Gamma
_{2}\left(  G|s_{0}\right)  $ is a zero-sum game, \ it follows that
$\widehat{\sigma}^{1},\widehat{\sigma}^{2}$ are \emph{optimal} and yield the
\emph{value} of the game. More precisely, we have the following.

\begin{theorem}
\label{prop0401}For every graph $G$ and every initial state $s_{0}\in S$, the
profile of deterministic positional strategies $\widehat{\sigma}=\left(
\widehat{\sigma}^{1},\widehat{\sigma}^{2}\right)  $ specified by Theorem
\ref{prop0301} satisfies%
\[
\max_{\sigma^{1}}\min_{\sigma^{2}}Q^{1}\left(  s_{0},\sigma^{1},\sigma
^{2}\right)  =Q^{1}\left(  s_{0},\widehat{\sigma}^{1},\widehat{\sigma}%
^{2}\right)  =\min_{\sigma^{2}}\max_{\sigma^{1}}Q^{1}\left(  s_{0},\sigma
^{1},\sigma^{2}\right)  .
\]

\end{theorem}

Furthermore, $\widehat{\sigma}^{1},\widehat{\sigma}^{2}$ and $Q^{n}\left(
s_{0},\widehat{\sigma}^{1},\widehat{\sigma}^{2}\right)  $ can be computed by a
\emph{value iteration algorithm} \cite{raghavanfilar1991}. Hence $\Gamma
_{2}\left(  G|s_{0}\right)  $ is completely solved.

Let us now discuss the connection of $\Gamma_{2}\left(  G|s_{0}\right)  $ to
the classic CR\ game. Note that the above description of $\Gamma_{2}\left(
G|s_{0}\right)  $ is almost identical to that of the \emph{time optimal
}version\emph{ }of the classic CR\ game (e.g., see \cite[Section
8.6]{bonato2011}). We only have the following differences.

\begin{enumerate}
\item In $\Gamma_{2}\left(  G|s_{0}\right)  $ time is measured in turns; in
classic CR it is measured in \emph{rounds}, where each round consists of one
$P_{1}$ turn and one $P_{2}$ turn.

\item In $\Gamma_{2}\left(  G|s_{0}\right)  $ the starting position $s_{0}$ is
\emph{given}; in classic CR it is chosen by the players, in an initial
\textquotedblleft\emph{placement}\textquotedblright\ round. In other words,
classic CR\ starts with an \textquotedblleft empty\textquotedblright\ graph;
in the first turn of the $0$-th round $P_{1}$ chooses his initial position; in
the second turn $P_{2}$, having observed $P_{1}$'s placement chooses his
initial position (after placement, classic CR is played exactly as $\Gamma
_{2}\left(  G|s_{0}\right)  $).
\end{enumerate}

\noindent At any rate, the important points are the following.

\begin{enumerate}
\item Having computed the values $u\left(  s_{0}\right)  $ of $\Gamma
_{2}\left(  G|s_{0}\right)  $ for every $s_{0}\in S$, we can easily obtain the
\emph{optimal capture time }$\widehat{T}_{C}$ of the classic CR
game\footnote{Up to a time rescaling, due to the abovementioned difference of
of time units.} as follows:%
\[
\widehat{T}_{C}=\frac{\log\left(  \max_{x^{1}}\min_{x^{2}}u^{1}\left(  \left(
x^{1},x^{2},1\right)  \right)  \right)  }{\log\gamma};
\]
furthermore, any $\widehat{x}^{1},\widehat{x}^{2}$ which satisfy $\widehat
{T}_{C}=\frac{\log u^{1}\left(  \left(  \widehat{x}^{1},\widehat{x}%
^{2},1\right)  \right)  }{\log\gamma}$ are optimal initial placements for
$P_{1}$ and $P_{2}$; and the optimal policies of $\Gamma_{2}\left(
G|s_{0}\right)  $ are time optimal policies (after placement)\ of the classic CR.

\item In the classic CR\ literature, a graph $G$ is called \emph{cop-win} iff
a single cop can capture the robber when both cop and robber play optimally on
$G$. In the more general case, where the cop player controls one or more cop
\emph{tokens}, the \emph{cop number} of $G$ is denoted by $c\left(  G\right)
$ and defined to be the smallest number of cop tokens which guarantees capture
when CR\ is played optimally on $G$. Clearly a graph is cop-win iff $c\left(
G\right)  =1$. It is easily seen that we can check whether $G$ is cop-win by
solving $\Gamma_{2}\left(  G|s_{0}\right)  $ (for all $s_{0}$) as indicated by
the following equivalence:
\begin{equation}
c\left(  G\right)  =1\Leftrightarrow\max_{x^{1}}\min_{x^{2}}u^{1}\left(
\left(  \widehat{x}^{1},\widehat{x}^{2},1\right)  \right)  >0.\label{eq0401}%
\end{equation}

\end{enumerate}

\noindent While the above questions regarding classic CR\ have been studied in
the related literature and answered using graph theoretic methods, the
connection to Game Theory appears to not have been previously exploited.

\section{GCR\ with Three Players\label{sec05}}

By substituting $N=3$ in the definitions of Section \ref{sec02} we obtain the
game $\Gamma_{3}\left(  G|s_{0}\right)  $; in particular we get the sets of
capture states%
\begin{align*}
\widetilde{S}^{1} &  =\left\{  s:\left(  x^{1},x^{2},x^{3},p\right)
,x^{1}=x^{2}\right\}  \text{ (}P_{1}\text{ captures }P_{2}\text{),}\\
\widetilde{S}^{2} &  =\left\{  s:\left(  x^{1},x^{2},x^{3},p\right)
,x^{2}=x^{3}\right\}  \text{ (}P_{2}\text{ captures }P_{3}\text{)}%
\end{align*}
and we use these to define the turn payoffs $q^{n}$ as follows%
\begin{equation}
q^{1}\left(  s\right)  =\left\{
\begin{array}
[c]{rll}%
1 & \text{iff} & s\in\widetilde{S}^{1},\\
0 & \text{else;} &
\end{array}
\right.  \quad q^{2}\left(  s\right)  =\left\{
\begin{array}
[c]{rll}%
-1 & \text{iff} & s\in\widetilde{S}^{1},\\
1 & \text{iff} & s\in\widetilde{S}^{2}\backslash\widetilde{S}^{1},\\
0 & \text{else;} &
\end{array}
\right.  \quad q^{3}\left(  s\right)  =\left\{
\begin{array}
[c]{rll}%
-1 & \text{iff} & s\in\widetilde{S}^{2}\backslash\widetilde{S}^{1},\\
0 & \text{else.} &
\end{array}
\right.  .\quad\label{eq0501a}%
\end{equation}
Note that, according to previous remarks, $P_{2}$ (resp. $P_{3}$) is rewarded
(resp. penalized)\ when $P_{2}$ captures $P_{3}$ \emph{and is not
simultaneously captured by }$P_{1}$. Also, recall that the total payoff
function is, as usual,%
\[
\forall n\in\left[  3\right]  :Q^{n}\left(  s_{0},s_{1},...\right)
=\sum_{t=0}^{\infty}\gamma^{t}q^{n}\left(  s_{t}\right)  .
\]
We are now ready to study $\Gamma_{3}\left(  G|s_{0}\right)  $.

\subsection{Nash Equilibria: Positional and Non-Positional\label{sec0501}}

By Theorem \ref{prop0301} we know that $\Gamma_{3}\left(  G|s_{0}\right)  $
has, for every $G$ and $s_{0}$, a NE in deterministic positional strategies.
In addition, as we will now show, $\Gamma_{3}\left(  G|s_{0}\right)  $ has at
least one NE in \emph{non-positional} deterministic strategies.

To this end we will introduce a family of \emph{auxiliary games} and
\emph{threat strategies} \cite{boros2009,chatterjee2003,thuijsman1997}. For
every $n\in\left[  3\right]  $\ we define the game $\widetilde{\Gamma}_{3}%
^{n}\left(  G|s_{0}\right)  $ played on $G$ (and starting at $s_{0}$) \ by
$P_{n}$ against a player $P_{-n}$ who controls the remaining two entities. For
example, in $\widetilde{\Gamma}_{3}^{1}\left(  G|s_{0}\right)  $, $P_{1}$
plays against $P_{-1}$ who controls $P_{2}$ and $P_{3}$. The $\widetilde
{\Gamma}_{3}^{n}\left(  G|s_{0}\right)  $ elements (e.g., movement sequence,
states, action sets, capturing conditions etc.) are the same as in $\Gamma
_{3}\left(  G|s_{0}\right)  $. $P_{n}$ uses a strategy $\sigma^{n}$ and
$P_{-n}$ uses a strategy profile $\sigma^{-n}$; these form a strategy profile
$\sigma=\left(  \sigma^{1},\sigma^{2},\sigma^{3}\right)  $ (which can also be
used in $\Gamma_{3}\left(  G|s_{0}\right)  $). The \ payoffs to $P_{n}$ and
$P_{-n}$ in $\widetilde{\Gamma}_{3}^{n}\left(  G|s_{0}\right)  $ are%
\[
\widetilde{Q}^{n}\left(  s_{0},\sigma\right)  =Q^{n}\left(  s_{0}%
,\sigma\right)  =\sum_{t=0}^{\infty}\gamma^{t}q^{n}\left(  s_{t}\right)
\text{\quad and\quad}\widetilde{Q}^{-n}\left(  s_{0},\sigma\right)
=-\widetilde{Q}^{n}\left(  s_{0},\sigma\right)  .
\]
Since the capture rules of $\widetilde{\Gamma}_{3}^{n}\left(  G|s_{0}\right)
$ are those of $\Gamma_{3}\left(  G|s_{0}\right)  $, $P_{-n}$ can use one of
his tokens to capture the other. For instance, in $\widetilde{\Gamma}_{3}%
^{1}\left(  G|s_{0}\right)  $, $P_{-1}$ can use $P_{2}$ to capture $P_{3}$ (as
will be seen in a later example, in certain cases this can be an optimal
move). Note however that in this case $P_{1}$ receives zero payoff (since he
did not capture) and $P_{-1}$ also receives zero payoff (since, by
construction, $\widetilde{\Gamma}_{3}^{1}\left(  G|s_{0}\right)  $, is a
zero-sum game).

In short, $\widetilde{\Gamma}_{3}^{n}\left(  G|s_{0}\right)  $ is a two-player
\emph{zero-sum }discounted stochastic game and the next Lemma follows from the
results of \cite[Theorem 4.3.2]{filar1996}.

\begin{lemma}
\label{prop0501}For every $n,G$ and $s_{0}$ the game $\widetilde{\Gamma}%
_{3}^{n}\left(  G|s_{0}\right)  $ has a value and the players have optimal
deterministic positional strategies.
\end{lemma}

Furthermore, the value and optimal strategies can be computed by Shapley's
value-iteration algorithm \cite{raghavanfilar1991}. Let us denote by
$\widehat{\phi}_{n}^{n}$ (resp. $\widehat{\phi}_{n}^{-n}$)\ the optimal
strategy of $P_{n}$ (resp. $P_{-n}$) in $\widetilde{\Gamma}_{3}^{n}\left(
G|s_{0}\right)  $. For example, in $\widetilde{\Gamma}_{3}^{1}\left(
G|s_{0}\right)  $, $P_{1}$ has the optimal strategy $\widehat{\phi}_{1}^{1}$
and $P_{-1}$ has the optimal strategy $\widehat{\phi}_{1}^{-1}$ $=\left(
\widehat{\phi}_{1}^{2},\widehat{\phi}_{1}^{3}\right)  $. In fact the same
$\widehat{\phi}_{n}^{m}$'s (for fixed $n$ and any $m\in\left[  3\right]  $)
are optimal in $\widetilde{\Gamma}_{3}^{n}\left(  G|s_{0}\right)  $ for every
initial position $s_{0}$.

We return to $\Gamma_{3}\left(  G|s_{0}\right)  $, and for each $P_{n}$ we
introduce the \emph{threat strategy} $\widehat{\pi}^{n}$ defined \ as follows:

\begin{enumerate}
\item as long as every player $P_{m}$ (with $m\neq n$) follows $\widehat{\phi
}_{m}^{m}$, $P_{n}$ follows $\widehat{\phi}_{n}^{n}$;

\item as soon as some player $P_{m}$ (with $m\neq n$)$\ $deviates from
$\widehat{\phi}_{m}^{m}$, $P_{n}$ switches to $\widehat{\phi}_{m}^{n}$ and
uses it for the rest of the game\footnote{Since $\Gamma_{3}\left(
G|s_{0}\right)  $ is a perfect information game, the deviation will be
detected immediately.}.
\end{enumerate}

Note that the $\widehat{\pi}^{n}$ strategies are \emph{not}\ positional. In
particular, the action of a player at time $t$ may be influenced by the action
(deviation) performed by another player at time $t-2$. However, as we will now
prove, $\left(  \widehat{\pi}^{1},\widehat{\pi}^{2},\widehat{\pi}^{3}\right)
$ is a (non-positional) NE\ in $\Gamma_{3}\left(  G|s_{0}\right)  $.

\begin{theorem}
\label{prop0502}For every $G,s_{0}$ and $\gamma$, we have:%
\begin{equation}
\forall n\in\left\{  1,2,3\right\}  ,\forall\pi^{n}:Q^{n}(s,\widehat{\pi}%
^{1},\widehat{\pi}^{2},\widehat{\pi}^{3})\geq Q^{n}(s,\pi^{n},\widehat{\pi
}^{-n}). \label{eq001}%
\end{equation}

\end{theorem}

\begin{proof}
We choose some initial state $s_{0}$ and fix it for the rest of the proof. Now
let us prove (\ref{eq001}) for the case $n=1$. In other words, we will show
that%
\begin{equation}
\forall\pi^{1}:Q^{1}(s_{0},\widehat{\pi}^{1},\widehat{\pi}^{2},\widehat{\pi
}^{3})\geq Q^{1}(s_{0},\pi^{1},\widehat{\pi}^{2},\widehat{\pi}^{3}).
\label{eq002a}%
\end{equation}
We take any $\pi^{1}$ and let
\begin{align*}
\text{the history produced by }(\widehat{\pi}^{1},\widehat{\pi}^{2}%
,\widehat{\pi}^{3})\text{ be }\widehat{\mathbf{s}}  &  =\widehat{s}%
_{0}\widehat{s}_{1}\widehat{s}_{2}...,\\
\text{ the history produced by }(\pi^{1},\widehat{\pi}^{2},\widehat{\pi}%
^{3})\text{ be }\widetilde{\mathbf{s}}  &  =\widetilde{s}_{0}\widetilde{s}%
_{1}\widetilde{s}_{2}...,
\end{align*}
(where $\widehat{s}_{0}=\widetilde{s}_{0}=s_{0}$). We define $T_{1}$ as the
earliest time in which $\pi^{1}$ and $\widehat{\pi}^{1}$ produce different
states:
\[
T_{1}=\min\left\{  t:\widetilde{s}_{t}\neq\widehat{s}_{t}\right\}  ,
\]
If $T_{1}=\infty$, then $\widetilde{\mathbf{s}}=\widehat{\mathbf{s}}$ and
\begin{equation}
Q^{1}(s,\widehat{\pi}^{1},\widehat{\pi}^{2},\widehat{\pi}^{3})=Q^{1}(s,\pi
^{1},\widehat{\pi}^{2},\widehat{\pi}^{3}). \label{eq002b}%
\end{equation}
If $T_{1}<\infty$, on the other hand, then $\widetilde{s}_{t}=\widehat{s}_{t}$
for every $t<T_{1}$ and we have%
\begin{align}
Q^{1}(s,\widehat{\pi}^{1},\widehat{\pi}^{2},\widehat{\pi}^{3})  &  =\sum
_{t=0}^{T_{1}-2}\gamma^{t}q^{1}\left(  \widehat{s}_{t}\right)  +\sum
_{t=T_{1}-1}^{\infty}\gamma^{t}q^{1}\left(  \widehat{s}_{t}\right)
=\sum_{t=0}^{T_{1}-2}\gamma^{t}q^{1}\left(  \widetilde{s}_{t}\right)
+\sum_{t=T_{1}-1}^{\infty}\gamma^{t}q^{1}\left(  \widehat{s}_{t}\right)
,\label{eq003a}\\
Q^{1}(s,\pi^{1},\widehat{\pi}^{2},\widehat{\pi}^{3})  &  =\sum_{t=0}^{T_{1}%
-2}\gamma^{t}q^{1}\left(  \widetilde{s}_{t}\right)  +\sum_{t=T_{1}-1}^{\infty
}\gamma^{t}q^{1}\left(  \widetilde{s}_{t}\right)  =\sum_{t=0}^{T_{1}-2}%
\gamma^{t}q^{1}\left(  \widetilde{s}_{t}\right)  +\sum_{t=T_{1}-1}^{\infty
}\gamma^{t}q^{1}\left(  \widetilde{s}_{t}\right)  . \label{eq004a}%
\end{align}
We define $s^{\ast}=\widehat{s}_{T_{1}-1}=\widetilde{s}_{T_{1}-1}$ and proceed
to compare the sums in (\ref{eq003a}) and (\ref{eq004a}).

First consider $\sum_{t=T_{1}-1}^{\infty}\gamma^{t}q^{1}\left(  \widehat
{s}_{t}\right)  $. The history $\widehat{\mathbf{s}}=\widehat{s}_{0}%
\widehat{s}_{1}\widehat{s}_{2}...$ is produced by $(\widehat{\phi}_{1}%
^{1}\ ,\widehat{\phi}_{2}^{2}\ ,\widehat{\phi}_{3}^{3}\ )$ and, since the
$\widehat{\phi}_{n}^{n}$'s are positional strategies, we have%
\begin{equation}
\sum_{t=T_{1}-1}^{\infty}\gamma^{t}q^{1}\left(  \widehat{s}_{t}\right)
=\gamma^{T_{1}-1}\sum_{t=0}^{\infty}\gamma^{t}q^{1}\left(  \widehat{s}%
_{T_{1}-1+t}\right)  =\gamma^{T_{1}-1}\widetilde{Q}^{1}\left(  s^{\ast
},\widehat{\phi}_{1}^{1},\widehat{\phi}_{2}^{2},\widehat{\phi}_{3}^{3}\right)
, \label{eq005}%
\end{equation}
i.e., up to the multiplicative constant $\gamma^{T_{1}-1}$, the sum in
(\ref{eq005}) is the payoff to $P_{1}$ in $\widetilde{\Gamma}_{3}^{1}\left(
G|s^{\ast}\right)  $, under the strategies $\widehat{\phi}_{1}^{1},\left(
\widehat{\phi}_{2}^{2},\widehat{\phi}_{3}^{3}\right)  $. Since $\widetilde
{\Gamma}_{3}^{1}\left(  G|s^{\ast}\right)  $ is a zero-sum game in which the
optimal response to $\widehat{\phi}_{1}^{1}$ is $\left(  \widehat{\phi}%
_{1}^{2},\widehat{\phi}_{1}^{3}\right)  $; hence we have%
\begin{equation}
\gamma^{T_{1}-1}\widetilde{Q}^{1}\left(  s^{\ast},\widehat{\phi}_{1}%
^{1},\widehat{\phi}_{2}^{2},\widehat{\phi}_{3}^{3}\right)  \geq\gamma
^{T_{1}-1}\widetilde{Q}^{1}\left(  s^{\ast},\widehat{\phi}_{1}^{1}%
,\widehat{\phi}_{1}^{2},\widehat{\phi}_{1}^{3}\right)  .
\end{equation}

Next consider $\sum_{t=T_{1}-1}^{\infty}\gamma^{t}q^{1}\left(  \widetilde
{s}_{t}\right)  $. The history $\widetilde{\mathbf{s}}=\widetilde{s}%
_{0}\widetilde{s}_{1}\widetilde{s}_{2}...$ is produced by $(\pi^{1}%
,\widehat{\phi}_{1}^{2},\widehat{\phi}_{1}^{3})$ and, since $\pi^{1}$ is not
necessarily positional, $\widetilde{s}_{T_{1}}\widetilde{s}_{T_{1}%
+1}\widetilde{s}_{T_{1}+2}...$ \ may depend on $\widetilde{s}_{0}\widetilde
{s}_{1}...\widetilde{s}_{T_{1}-2}$. However, we can introduce a (not
necessarily positional) strategy $\rho^{1}$ \emph{ }which will produce the
same history $\widetilde{s}_{T_{1}}\widetilde{s}_{T_{1}+1}\widetilde{s}%
_{T_{1}+2}...$ \ as $\sigma^{1}$. \footnote{We define $\rho^{1}$ such that,
when combined with $\widetilde{s}_{T_{1}-1},\widehat{\phi}_{1}^{2}%
,\widehat{\phi}_{1}^{3}$, will produce the same history $\widetilde{s}_{T_{1}%
}\widetilde{s}_{T_{1}+1}\widetilde{s}_{T_{1}+2}...$ \ as $\sigma^{1}$. Note
that $\rho^{1}$ will in general depend (in an indirect way) on $\widetilde
{s}_{0}\widetilde{s}_{1}...\widetilde{s}_{T_{1}-2}$.} Then, since in
$\widetilde{\Gamma}_{3}^{1}\left(  G|s^{\ast}\right)  $ the optimal response
to $\left(  \widehat{\phi}_{1}^{2},\widehat{\phi}_{1}^{3}\right)  $ is
$\widehat{\phi}_{1}^{1}$, we have%
\begin{equation}
\gamma^{T_{1}-1}\widetilde{Q}^{1}\left(  s^{\ast},\widehat{\phi}_{1}%
^{1},\widehat{\phi}_{1}^{2},\widehat{\phi}_{1}^{3}\right)  \geq\gamma
^{T_{1}-1}\widetilde{Q}^{1}\left(  s^{\ast},\rho^{1},\widehat{\phi}_{1}%
^{2},\widehat{\phi}_{1}^{3}\right)  =\sum_{t=T_{1}-1}^{\infty}\gamma^{t}%
q^{1}\left(  \widetilde{s}_{t}\right)  . \label{eq006a}%
\end{equation}

Combining (\ref{eq003a})-(\ref{eq006a}) we have:%
\begin{align*}
Q^{1}(s_{0},\widehat{\pi}^{1},\widehat{\pi}^{2},\widehat{\pi}^{3})  &
=\sum_{t=0}^{T_{1}-2}\gamma^{t}q^{1}\left(  \widetilde{s}_{t}\right)
+\gamma^{T_{1}-1}\widetilde{Q}^{1}(s^{\ast},\widehat{\phi}_{1}^{1}%
,\widehat{\phi}_{2}^{2},\widehat{\phi}_{3}^{3})\\
&  \geq\sum_{t=0}^{T_{1}-2}\gamma^{t}q^{1}\left(  \widetilde{s}_{t}\right)
+\gamma^{T_{1}-1}\widetilde{Q}^{1}(s^{\ast},\widehat{\phi}_{1}^{1}%
,\widehat{\phi}_{1}^{2},\widehat{\phi}_{1}^{3})\\
&  \geq\sum_{t=0}^{T_{1}-2}\gamma^{t}q^{1}\left(  \widetilde{s}_{t}\right)
+\gamma^{T_{1}-1}\widetilde{Q}^{1}(s^{\ast},\rho^{1},\widehat{\phi}_{1}%
^{2},\widehat{\phi}_{1}^{3})=Q^{1}(s,\pi^{1},\widehat{\pi}^{2},\widehat{\pi
}^{3}).
\end{align*}
and we have proved (\ref{eq002a}), which is (\ref{eq001})\ for $n=1$. The
proof for the cases $n=2$ and $n=3$ is similar and hence omitted.
\end{proof}

We have seen that every $\Gamma_{3}\left(  G|s_{0}\right)  $ has at least two
deterministic NE (one in positional strategies and another in \emph{non}%
-positional ones); and in fact, as is well known, a stochastic game may
possess any number of NE. On the other hand, we only know how to compute a
single NE of $\Gamma_{3}\left(  G|s_{0}\right)  $, namely the non-positional
one of Theorem \ref{prop0502}, which is constructed in terms of the two-player
strategies of $\widetilde{\Gamma}_{3}^{n}\left(  G|s_{0}\right)  $. One may be
tempted to construct additional NE$\ $of $\Gamma_{3}\left(  G|s_{0}\right)  $
using the optimal strategies of $\Gamma_{2}\left(  G|s_{0}\right)  $. For
example, one may reason as follows: $P_{3}$'s best chance to avoid capture in
$\Gamma_{3}\left(  G|s_{0}\right)  $ is by ignoring $P_{1}$ and playing his
best (in $\Gamma_{2}\left(  G|s_{0}\right)  $) evasion strategy against
$P_{2}$. By a similar reasoning for the other players, one may conclude that,
$\left(  \widehat{\sigma}^{1},\widehat{\sigma}^{2},\widehat{\sigma}%
^{3}\right)  $ is a NE\ of $\Gamma_{3}\left(  G|s_{0}\right)  $ if
(i)$\ \left(  \widehat{\sigma}^{1},\widehat{\sigma}^{2}\right)  $ is a NE\ of
$\Gamma_{2}\left(  G|s_{0}\right)  $ played between $P_{1}$ and $P_{2}$, and
(ii)$\ \left(  \widehat{\sigma}^{2},\widehat{\sigma}^{3}\right)  $ is a NE\ of
$\Gamma_{2}\left(  G|s_{0}\right)  $ played between $P_{2}$ and $P_{3}$.
\footnote{A clarification is needed here: the domain of $\Gamma_{2}\left(
G|s_{0}\right)  $ (positional)\ strategies is $V\times V\times\left\{
1,2\right\}  $, while the the domain of $\Gamma_{3}\left(  G|s_{0}\right)  $
(positional)\ strategies is $V\times V\times V\times\left\{  1,2\right\}  $.
However we can \textquotedblleft extend\textquotedblright\ a $\Gamma
_{2}\left(  G|s_{0}\right)  $ strategy to use it in $\Gamma_{3}\left(
G|s_{0}\right)  $. For example,\ suppose $\sigma^{1}\left(  x^{1}%
,x^{2}\right)  $ is a $P_{1}$ strategy in $\Gamma_{2}\left(  G|s_{0}\right)
$; then it can also be extended to a $\Gamma_{3}\left(  G|s_{0}\right)  $
strategy $\widetilde{\sigma}^{1}\left(  x^{1},x^{2},x^{3}\right)  $ by letting%
\[
\forall x^{3}:\widetilde{\sigma}^{1}\left(  x^{1},x^{2},x^{3}\right)
=\sigma^{1}\left(  x^{1},x^{2}\right)  .
\]
In other words, $P_{1}$ applies $\sigma^{1}$ in $\Gamma_{3}\left(
G|s_{0}\right)  $ by ignoring $P_{3}$'s position. We will often use this and
similar constructions in what follows, without further comment; and we will
denote the $\Gamma_{2}\left(  G|s_{0}\right)  $ and $\Gamma_{3}\left(
G|s_{0}\right)  $ strategies by the same symbol, e.g., $\sigma^{n}$.} But this
conclusion is wrong, as shown by the following example.

\begin{example}
\label{prop0503}\normalfont Consider $\Gamma_{2}\left(  G|s_{0}\right)  $ when
$G$ is the graph of Figure \ref{fig0501} and $s_{0}=\left(  12,2,3\right)  $,
as indicated in the figure (for the time being suppose $P_{1}$ is not on the
graph). $P_{3}$ is the evader and his best strategy is to move towards vertex
10, postponing capture as long as possible; $P_{2}$ is the pursuer and his
best strategy is to always move toward $P_{3}$. Now consider $\Gamma
_{3}\left(  G|s_{0}\right)  $ with $s_{0}=\left(  1,12,2,3\right)  $. In this
game, $P_{3}$'s best strategy is to first move into vertex 1 and afterwards
always keep $P_{1}$ between himself and $P_{2}$; he can always achieve this
and thus avoid capture ad infinitum. And $P_{2}$'s best strategy is to stay at
vertex 12, keeping away from $P_{1}$ for as long as possible. So in this
example $P_{2}$ and $P_{3}$'s optimal $\Gamma_{2}\left(  G|s_{0}\right)  $
strategies are not good (and certainly not in NE)\ in $\Gamma_{3}\left(
G|s_{0}\right)  $.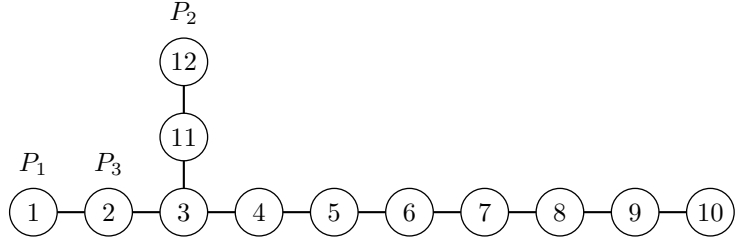
\begin{figure}[ptbh]
\begin{center}
\begin{tikzpicture}
\SetGraphUnit{2}
\Vertex[x= 0,y= 0]{1}
\Vertex[x= 1,y= 0]{2}
\Vertex[x= 2,y= 0]{3}
\Vertex[x= 3,y= 0]{4}
\Vertex[x= 4,y= 0]{5}
\Vertex[x= 5,y= 0]{6}
\Vertex[x= 6,y= 0]{7}
\Vertex[x= 7,y= 0]{8}
\Vertex[x= 8,y= 0]{9}
\Vertex[x= 9,y= 0]{10}
\Vertex[x= 2,y= 1]{11}
\Vertex[x= 2,y= 2]{12}
\node(A) [label=$P_1$] at (0,0.25) {};
\node(B) [label=$P_3$] at (1,0.25) {};
\node(C) [label=$P_2$] at (2,2.25) {};
\Edge(1)(2)
\Edge(2)(3)
\Edge(3)(4)
\Edge(4)(5)
\Edge(5)(6)
\Edge(6)(7)
\Edge(7)(8)
\Edge(8)(9)
\Edge(9)(10)
\Edge(3)(11)
\Edge(11)(12)
\SetVertexNoLabel
\end{tikzpicture}
\end{center}
\par
\label{fig0501}\caption{A case in which the CR optimal strategies do not
achieve NE in $\Gamma_{3}(G|s_{0})$.}%
\end{figure}
\end{example}

\subsection{Capturability\label{sec0502}}

In Section \ref{sec04} we have presented a connection between the cop number
of $G$ and \textquotedblleft capturability\textquotedblright\ in $\Gamma
_{2}\left(  G|s_{0}\right)  $; this was described by (\ref{eq0401}) which can
be equivalently rewritten as%
\begin{equation}
\left(  \forall s_{0}\text{ every optimal }\widehat{\sigma}\text{ of }%
\Gamma_{2}\left(  G|s_{0}\right)  \ \text{ results in capture}\right)
\Leftrightarrow c\left(  G\right)  =1. \label{eq0509}%
\end{equation}
The analog of (\ref{eq0509}) in $\Gamma_{3}\left(  G|s_{0}\right)  $ would
be:
\begin{equation}
\left(  \forall s_{0}\text{ every NE }\widehat{\sigma}\text{ of }\Gamma
_{3}\left(  G|s_{0}\right)  \text{ results in capture}\right)  \Leftrightarrow
c\left(  G\right)  =1. \label{eq0510}%
\end{equation}
As will be seen, \emph{(\ref{eq0510}) is not true}. But connections between
cop number and capturability exist, as will be established in the remainder of
this section. To this end, we first define the capture function $\mathbf{K}%
_{3}\left(  G|s_{0},\sigma\right)  $.

\begin{definition}
\label{prop0504}For the game $\Gamma_{3}\left(  G|s_{0}\right)  $ played with
strategies $\sigma=\left(  \sigma^{1},\sigma^{2},\sigma^{3}\right)  $, we
write
\[
\mathbf{K}_{3}\left(  G|s_{0},\sigma\right)  =\left\{
\begin{array}
[c]{lll}%
0 & \text{when} & Q^{1}\left(  s_{0},\sigma\right)  =Q^{2}\left(  s_{0}%
,\sigma\right)  =Q^{3}\left(  s_{0},\sigma\right)  =0,\\
1 & \text{when} & Q^{1}\left(  s_{0},\sigma\right)  >0,\\
2 & \text{when} & Q^{2}\left(  s_{0},\sigma\right)  >0.
\end{array}
\right.
\]

\end{definition}

\noindent Roughly, $\mathbf{K}_{3}\left(  G|s_{0},\sigma\right)  $ tells us
which player (if any)\ achieves a capture in $\Gamma_{3}\left(  G|s_{0}%
\right)  $ played with $\left(  \sigma^{1},\sigma^{2},\sigma^{3}\right)  $:

\begin{enumerate}
\item $\mathbf{K}_{3}\left(  G|s_{0},\sigma\right)  =0\Leftrightarrow$
$Q^{1}\left(  s_{0},\sigma\right)  =Q^{2}\left(  s_{0},\sigma\right)
=Q^{3}\left(  s_{0},\sigma\right)  =0$ $\Leftrightarrow$ no capture takes place;

\item $\mathbf{K}_{3}\left(  G|s_{0},\sigma\right)  =1\Leftrightarrow$
$Q^{1}\left(  s_{0},\sigma\right)  >0$ $\Leftrightarrow P_{1}$ captures
$P_{2}$;

\item $\mathbf{K}_{3}\left(  G|s_{0},\sigma\right)  =2\Leftrightarrow
\ Q^{2}\left(  s_{0},\sigma\right)  >0$ $\Leftrightarrow P_{2}$ captures
$P_{3}$ (and avoids being captured by $P_{1}$).
\end{enumerate}

\noindent A weaker version of (\ref{eq0510}) is:%
\[
\left(  \forall s_{0}\text{ there exists a capturing NE }\widehat{\sigma
}\text{ of }\Gamma_{3}\left(  G|s_{0}\right)  \right)  \Rightarrow c\left(
G\right)  =1
\]

\noindent and this can be rewritten and proved in terms of $\mathbf{K}%
_{3}\left(  G|s_{0},\sigma\right)  $, as follows.

\begin{theorem}
\label{prop0505}The following holds for every $G$:%
\begin{equation}
\left(  \forall s_{0}\text{ there exists a NE }\widehat{\sigma}\text{ of
}\Gamma_{3}\left(  G|s_{0}\right)  :\text{ }\mathbf{K}_{3}\left(
G|s_{0},\widehat{\sigma}\right)  >0\right)  \Rightarrow c\left(  G\right)  =1.
\label{eq0511}%
\end{equation}

\end{theorem}

\begin{proof}
To prove the theorem we will assume
\begin{equation}
\left(  \forall s_{0}\text{ there exists a NE }\widehat{\sigma}\text{ of
}\Gamma_{3}\left(  G|s_{0}\right)  :\text{ }\mathbf{K}_{3}\left(
G|s_{0},\widehat{\sigma}\right)  >0\right)  \text{ and }c\left(  G\right)  >1
\end{equation}
\noindent and reach a contradiction. To this end choose $s_{0}=\left(
x^{1},x^{2},x^{3},1\right)  $ as follows.

\begin{enumerate}
\item Take arbitrary $x^{1}$.

\item Take some $x^{2}$ such that for $s_{0}^{\prime}=\left(  x^{1}%
,x^{2},1\right)  $ there exists a $\overline{\sigma}^{2}$ which is escaping in
$\Gamma_{2}\left(  G|s_{0}^{\prime}\right)  $ (this is always possible, since
$c\left(  G\right)  >1$).

\item Take some $x^{3}$ such that for $s_{0}^{\prime\prime}=\left(
x^{2},x^{3},1\right)  $ there exists a $\overline{\sigma}^{3}$ which is
escaping in $\Gamma_{2}\left(  G|s_{0}^{\prime\prime}\right)  $ (this is
always possible, since $c\left(  G\right)  >1$).
\end{enumerate}

\noindent Now let $\widehat{\sigma}=\left(  \widehat{\sigma}^{1}%
,\widehat{\sigma}^{2},\widehat{\sigma}^{3}\right)  $ be a capturing NE of
$\Gamma_{3}\left(  G|s_{0}\right)  $ and consider the following cases.

\begin{enumerate}
\item $\mathbf{K}_{3}\left(  G|s_{0},\widehat{\sigma}\right)  =1$; then, for
some $T_{1}$ we will have%
\[
Q^{2}\left(  s_{0},\widehat{\sigma}^{1},\widehat{\sigma}^{2},\widehat{\sigma
}^{3}\right)  =-\gamma^{T_{1}}<0\leq Q^{2}\left(  s_{0},\widehat{\sigma}%
^{1},\overline{\sigma}^{2},\widehat{\sigma}^{3}\right)  .
\]
Because, when $P_{1}$ and $P_{2}$ play $\widehat{\sigma}^{1}$ and
$\overline{\sigma}^{2}$, respectively, $P_{2}$ will always escape $P_{1}$
(since $P_{3}$ can never influence $P_{2}$ moves). And furthermore, $P_{2}$
may in fact capture $P_{3}$, since $\widehat{\sigma}^{3}$ is not necessarily
an escaping strategy.

\item $\mathbf{K}_{3}\left(  G|s_{0},\widehat{\sigma}\right)  =2$; then, for
some $T_{2}$ we will have%
\[
Q^{3}\left(  s_{0},\widehat{\sigma}^{1},\widehat{\sigma}^{2},\widehat{\sigma
}^{3}\right)  =-\gamma^{T_{2}}<0=Q^{3}\left(  s_{0},\widehat{\sigma}%
^{1},\widehat{\sigma}^{2},\overline{\sigma}^{3}\right)
\]
Because, when $P_{2}$ and $P_{3}$ play $\widehat{\sigma}^{2}$ and
$\overline{\sigma}^{3}$, respectively, $P_{3}$ will always escape $P_{2}$
(since $P_{1}$ can never influence $P_{3}$ moves).
\end{enumerate}

\noindent Since in every case some $P_{n}$ can unilaterally improve
$Q^{n}\left(  s_{0},\widehat{\sigma}\right)  $, $\widehat{\sigma}$ cannot be a
NE\ of $\Gamma_{3}\left(  G|s_{0}\right)  $.
\end{proof}

Hence for every cop-win graph $G$ and every starting state $s_{0}$,
$\Gamma_{3}\left(  G|s_{0}\right)  $ has a capturing NE. However, perhaps
surprisingly, there exists cop-win graphs and starting states for which
$\Gamma_{3}\left(  G|s_{0}\right)  $ also has noncapturing NE, as the
following example shows.

\begin{example}
\label{prop0506}\normalfont Take a path with $P_{1}$ and $P_{2}$ at the
endpoints and $P_{3}$ at the middle, as shown in Figure \ref{fig0502}.
\begin{figure}[ptbh]
\begin{center}
\begin{tikzpicture}
\SetGraphUnit{2}
\Vertex[x= 0,y= 0]{1}
\Vertex[x= 1,y= 0]{2}
\Vertex[x= 2,y= 0]{3}
\Vertex[x= 3,y= 0]{4}
\Vertex[x= 4,y= 0]{5}
\node(A) [label=$P_1$] at (0,0.25) {};
\node(B) [label=$P_3$] at (2,0.25) {};
\node(C) [label=$P_2$] at (4,0.25) {};
\Edge(1)(2)
\Edge(2)(3)
\Edge(3)(4)
\Edge(4)(5)
\SetVertexNoLabel
\end{tikzpicture}
\end{center}
\caption{A graph $G$ in which $\Gamma_{3}\left(  G|s_{0}\right)  $ has a
noncapturing NE.}%
\label{fig0502}%
\end{figure}
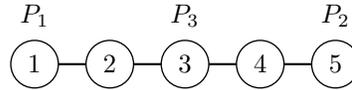The strategy profile $\overline{\sigma}=\left(  \overline{\sigma
}^{1},\overline{\sigma}^{2},\overline{\sigma}^{3}\right)  $ is defined as follows.
\end{example}

\begin{enumerate}
\item $\overline{\sigma}^{1}$: $P_{1}$ stays in place as long as $P_{2}$ does
not move; if $P_{2}$ moves, $P_{1}$ chases him.

\item $\overline{\sigma}^{2}$: $P_{2}$ stays in place as long as $P_{3}$ does
not move; if $P_{3}$ moves, $P_{2}$ chases him.

\item $\overline{\sigma}^{3}$: $P_{3}$ stays in place as long as nobody moves;
if $P_{1}$ moves, $P_{3}$ goes towards $P_{2}$; if $P_{2}$ moves, $P_{3}$ goes
towards $P_{1}$.
\end{enumerate}

\noindent We will now show $\overline{\sigma}$ is a noncapturing NE of
$\Gamma\left(  G|s_{0}\right)  $. Obviously we have%
\[
\forall n\in\left\{  1,2,3\right\}  :Q^{n}\left(  s_{0},\left(  \overline
{\sigma}^{1},\overline{\sigma}^{2},\overline{\sigma}^{3}\right)  \right)  =0.
\]
We will show no player profits by unilaterally changing his strategy.

\begin{enumerate}
\item Say $P_{1}$ uses any strategy $\sigma^{1}$. If, by $\sigma^{1}$, he
moves at some time, then $P_{3}$ goes towards $P_{2}$ and $P_{2}$ goes towards
$P_{3}$ resulting in a capture of $P_{3}$ by $P_{2}$. Hence
\begin{equation}
Q^{1}\left(  s_{0},\left(  \sigma^{1},\overline{\sigma}^{2},\overline{\sigma
}^{3}\right)  \right)  =0=Q^{1}\left(  s_{0},\left(  \overline{\sigma}%
^{1},\overline{\sigma}^{2},\overline{\sigma}^{3}\right)  \right)  .
\label{eq0501}%
\end{equation}

\item Say $P_{2}$ uses any strategy $\sigma^{2}$. If, by $\sigma^{2}$, he
moves at some time, then $P_{3}$ goes towards $P_{1}$ and $P_{1}$ goes towards
$P_{2}$ resulting in a capture of $P_{2}$ by $P_{1}$. Hence
\begin{equation}
Q^{2}\left(  s_{0},\left(  \overline{\sigma}^{1},\sigma^{2},\overline{\sigma
}^{3}\right)  \right)  <0=Q^{2}\left(  s_{0},\left(  \overline{\sigma}%
^{1},\overline{\sigma}^{2},\overline{\sigma}^{3}\right)  \right)  .
\label{eq0502}%
\end{equation}

\item Say $P_{3}$ uses any strategy $\sigma^{3}$. If, by $\sigma^{3}$, he
moves at some time, then $P_{2}$ goes towards $P_{3}$ and then $P_{1}$ goes
towards $P_{2}$. Depending on $P_{3}$'s moves we may have a capture of $P_{2}$
by $P_{1}$ or of $P_{3}$ by $P_{2}$. In either case
\begin{equation}
Q^{3}\left(  s_{0},\left(  \overline{\sigma}^{1},\overline{\sigma}^{2}%
,\sigma^{3}\right)  \right)  \leq0=Q^{3}\left(  s_{0},\left(  \overline
{\sigma}^{1},\overline{\sigma}^{2},\overline{\sigma}^{3}\right)  \right)  .
\label{eq0503}%
\end{equation}

\end{enumerate}

\noindent Combining (\ref{eq0501})-(\ref{eq0503}) we get%
\[
\forall n\in\left\{  1,2,3\right\}  ,\forall\sigma^{n}:Q^{n}\left(
s_{0},\left(  \sigma^{n},\overline{\sigma}^{-n}\right)  \right)  \leq
Q^{n}\left(  s_{0},\overline{\sigma}\right)
\]
which shows that $\overline{\sigma}$ is a noncapturing NE\ of $\Gamma\left(
G|s_{0}\right)  $.

\normalfont The above example shows that the converse of Theorem
\ref{prop0505} does not hold, i.e., there exist cop-win graphs $G$ and initial
states $s_{0}$ such that $\Gamma_{3}\left(  G|s_{0}\right)  $ has noncapturing
NE. However we can prove a weaker result:\ the converse does hold when $G$ is
a tree.

The first step in our proof is to revisit the two-player game $\widetilde
{\Gamma}_{3}^{2}\left(  G|s_{0}\right)  $ of Section \ref{sec0501}. Recall
that it is played between $P_{2}$ and $P_{-2}$ who controls the \emph{tokens}
$P_{1}$ and $P_{3}$. We now prove the following.

\begin{theorem}
\label{prop0507}If $c\left(  G\right)  =1$ then every optimal profile
$\widehat{\sigma}=\left(  \widehat{\sigma}^{1},\widehat{\sigma}^{2}%
,\widehat{\sigma}^{3}\right)  $ of $\widetilde{\Gamma}_{3}^{2}\left(
G|s_{0}\right)  $ is a NE of $\Gamma_{3}\left(  G|s_{0}\right)  $.
\end{theorem}

\begin{proof}
Let us choose some initial state $s_{0}$ and some optimal (in $\widetilde
{\Gamma}_{3}^{2}\left(  G|s_{0}\right)  $) profile $\widehat{\sigma}=\left(
\widehat{\sigma}^{1},\widehat{\sigma}^{2},\widehat{\sigma}^{3}\right)  $, and
keep them fixed for the rest of the proof.

For any $\sigma=\left(  \sigma^{1},\sigma^{2},\sigma^{3}\right)  $, the
capture function $\mathbf{K}_{3}\left(  G|s_{0},\sigma\right)  $ will take a
value in $\left\{  0,1,2\right\}  $. The values correspond to three outcomes
in $\Gamma_{3}\left(  G|s_{0}\right)  $ and the same outcomes are obtained in
$\widetilde{\Gamma}_{3}^{2}\left(  G|s_{0}\right)  $ (the two games differ in
their payoffs but are played by the same rules):

\begin{enumerate}
\item $\mathbf{K}_{3}\left(  G|s_{0},\sigma\right)  =1$ means $P_{1}$ captures
$P_{2}$;

\item $\mathbf{K}_{3}\left(  G|s_{0},\sigma\right)  =2$ means $P_{2}$ captures
$P_{3}$ (and is not captured by $P_{1}$);

\item $\mathbf{K}_{3}\left(  G|s_{0},\sigma\right)  =0$ means neither $P_{2}$
nor $P_{3}$ is captured.
\end{enumerate}

\noindent So we will consider the three mutually exclusive cases
separately.\medskip\medskip

\noindent\underline{\textbf{I. }$\mathbf{K}_{3}\left(  G|s_{0},\left(
\widehat{\sigma}^{1},\widehat{\sigma}^{2},\widehat{\sigma}^{3}\right)
\right)  =1$ }. Let us examine each player's payoff.

\begin{enumerate}
\item For all $\left(  \sigma^{1},\sigma^{3}\right)  $ such that
$\mathbf{K}_{3}\left(  G|s_{0},\left(  \sigma^{1},\widehat{\sigma}^{2}%
,\sigma^{3}\right)  \right)  =1$, $\left(  \widehat{\sigma}^{1},\widehat
{\sigma}^{2},\widehat{\sigma}^{3}\right)  $ optimality in $\widetilde{\Gamma
}_{3}^{2}\left(  G|s_{0}\right)  $, implies%
\[
Q^{1}\left(  s_{0},\widehat{\sigma}^{1},\widehat{\sigma}^{2},\widehat{\sigma
}^{3}\right)  =\widetilde{Q}^{-2}\left(  s_{0},\widehat{\sigma}^{1}%
,\widehat{\sigma}^{2},\widehat{\sigma}^{3}\right)  \geq\widetilde{Q}%
^{-2}\left(  s_{0},\sigma^{1},\widehat{\sigma}^{2},\sigma^{3}\right)
=Q^{1}\left(  s_{0},\sigma^{1},\widehat{\sigma}^{2},\sigma^{3}\right)  .
\]
And for all $\left(  \sigma^{1},\sigma^{3}\right)  $ such that $\mathbf{K}%
_{3}\left(  G|s_{0},\left(  \sigma^{1},\widehat{\sigma}^{2},\sigma^{3}\right)
\right)  \neq1$ we have
\[
Q^{1}\left(  s_{0},\widehat{\sigma}^{1},\widehat{\sigma}^{2},\widehat{\sigma
}^{3}\right)  >0=Q^{1}\left(  s_{0},\sigma^{1},\widehat{\sigma}^{2},\sigma
^{3}\right)  .
\]
Hence%
\begin{equation}
\forall\sigma^{1},\sigma^{3}:Q^{1}\left(  s_{0},\widehat{\sigma}^{1}%
,\widehat{\sigma}^{2},\widehat{\sigma}^{3}\right)  \geq Q^{1}\left(
s_{0},\sigma^{1},\widehat{\sigma}^{2},\sigma^{3}\right)  \Rightarrow
\forall\sigma^{1}:Q^{1}\left(  s_{0},\widehat{\sigma}^{1},\widehat{\sigma}%
^{2},\widehat{\sigma}^{3}\right)  \geq Q^{1}\left(  s_{0},\sigma^{1}%
,\widehat{\sigma}^{2},\widehat{\sigma}^{3}\right)  \label{eqNE101}%
\end{equation}

\item For all $\left(  \sigma^{1},\sigma^{2},\sigma^{3}\right)  $ we have%
\[
Q^{2}\left(  s_{0},\sigma^{1},\sigma^{2},\sigma^{3}\right)  =\widetilde{Q}%
^{2}\left(  s_{0},\sigma^{1},\sigma^{2},\sigma^{3}\right)  ;
\]
combining with optimality in $\widetilde{\Gamma}_{3}^{2}\left(  G|s_{0}%
\right)  $ we get%
\begin{equation}
\forall\sigma^{2}:Q^{2}\left(  s_{0},\widehat{\sigma}^{1},\widehat{\sigma}%
^{2},\widehat{\sigma}^{3}\right)  =\widetilde{Q}^{2}\left(  s_{0}%
,\widehat{\sigma}^{1},\widehat{\sigma}^{2},\widehat{\sigma}^{3}\right)
\geq\widetilde{Q}^{2}\left(  s_{0},\widehat{\sigma}^{1},\sigma^{2}%
,\widehat{\sigma}^{3}\right)  =Q^{2}\left(  s_{0},\widehat{\sigma}^{1}%
,\sigma^{2},\widehat{\sigma}^{3}\right)  . \label{eqNE102}%
\end{equation}

\item And finally
\begin{equation}
\forall\sigma^{1},\sigma^{3}:Q^{3}\left(  s_{0},\widehat{\sigma}^{1}%
,\widehat{\sigma}^{2},\widehat{\sigma}^{3}\right)  =0\geq Q^{3}\left(
s_{0},\sigma^{1},\widehat{\sigma}^{2},\sigma^{3}\right)  \Rightarrow
\forall\sigma^{3}:Q^{3}\left(  s_{0},\widehat{\sigma}^{1},\widehat{\sigma}%
^{2},\widehat{\sigma}^{3}\right)  =0\geq Q^{3}\left(  s_{0},\widehat{\sigma
}^{1},\widehat{\sigma}^{2},\sigma^{3}\right)  . \label{eqNE103}%
\end{equation}

\end{enumerate}

Combining (\ref{eqNE101})-(\ref{eqNE103}), we see that \textbf{ }%
\begin{equation}
\mathbf{K}_{3}\left(  G|s_{0},\left(  \widehat{\sigma}^{1},\widehat{\sigma
}^{2},\widehat{\sigma}^{3}\right)  \right)  =1\Rightarrow\forall
n,\forall\sigma^{n}:Q^{n}\left(  s_{0},\widehat{\sigma}\right)  \geq
Q^{n}\left(  s_{0},\sigma^{n},\widehat{\sigma}^{-n}\right)  . \label{eqNE141}%
\end{equation}

\noindent\underline{\textbf{II. }$\mathbf{K}_{3}\left(  G|s_{0},\left(
\widehat{\sigma}^{1},\widehat{\sigma}^{2},\widehat{\sigma}^{3}\right)
\right)  =2$}. Then the following hold.

\begin{enumerate}
\item From $\left(  \widehat{\sigma}^{1},\widehat{\sigma}^{3}\right)  $
optimality we have%
\begin{equation}
\forall\sigma^{1},\sigma^{3}:0=Q^{1}\left(  s_{0},\widehat{\sigma}%
^{1},\widehat{\sigma}^{2},\widehat{\sigma}^{3}\right)  >\widetilde{Q}%
^{-2}\left(  s_{0},\widehat{\sigma}^{1},\widehat{\sigma}^{2},\widehat{\sigma
}^{3}\right)  \geq\widetilde{Q}^{-2}\left(  s_{0},\sigma^{1},\widehat{\sigma
}^{2},\sigma^{3}\right)  . \label{eqNE105}%
\end{equation}
We cannot have $\mathbf{K}_{3}\left(  G|s_{0},\left(  \sigma^{1}%
,\widehat{\sigma}^{2},\sigma^{3}\right)  \right)  =1$, because then we would
also have $\widetilde{Q}^{-2}\left(  s_{0},\sigma^{1},\widehat{\sigma}%
^{2},\sigma^{3}\right)  >0$, which contradicts (\ref{eqNE105}). If we have
either $\mathbf{K}_{3}\left(  G|s_{0},\left(  \sigma^{1},\widehat{\sigma}%
^{2},\sigma^{3}\right)  \right)  =2$ or $\mathbf{K}_{3}\left(  G|s_{0},\left(
\sigma^{1},\widehat{\sigma}^{2},\sigma^{3}\right)  \right)  =0$ then
\[
Q^{1}\left(  s_{0},\widehat{\sigma}^{1},\widehat{\sigma}^{2},\widehat{\sigma
}^{3}\right)  =0=Q^{1}\left(  s_{0},\sigma^{1},\widehat{\sigma}^{2},\sigma
^{3}\right)  .
\]
In short%
\begin{equation}
\forall\sigma^{1},\sigma^{3}:Q^{1}\left(  s_{0},\widehat{\sigma}^{1}%
,\widehat{\sigma}^{2},\widehat{\sigma}^{3}\right)  =0=Q^{1}\left(
s_{0},\sigma^{1},\widehat{\sigma}^{2},\sigma^{3}\right)  \Rightarrow
\forall\sigma^{1}:Q^{1}\left(  s_{0},\widehat{\sigma}^{1},\widehat{\sigma}%
^{2},\widehat{\sigma}^{3}\right)  =Q^{1}\left(  s_{0},\sigma^{1}%
,\widehat{\sigma}^{2},\widehat{\sigma}^{3}\right)  \label{eqNE111}%
\end{equation}

\item By the same argument as in the previous case we get
\begin{equation}
\forall\sigma^{2}:Q^{2}\left(  s_{0},\widehat{\sigma}^{1},\widehat{\sigma}%
^{2},\widehat{\sigma}^{3}\right)  \geq Q^{2}\left(  s_{0},\widehat{\sigma}%
^{1},\sigma^{2},\widehat{\sigma}^{3}\right)  . \label{eqNE112}%
\end{equation}

\item Finally%
\begin{equation}
\forall\sigma^{1},\sigma^{3}:0>Q^{3}\left(  s_{0},\widehat{\sigma}%
^{1},\widehat{\sigma}^{2},\widehat{\sigma}^{3}\right)  =\widetilde{Q}%
^{-2}\left(  s_{0},\widehat{\sigma}^{1},\widehat{\sigma}^{2},\widehat{\sigma
}^{3}\right)  \geq\widetilde{Q}^{-2}\left(  s_{0},\sigma^{1},\widehat{\sigma
}^{2},\sigma^{3}\right)  . \label{eqNE113}%
\end{equation}
We cannot have $\mathbf{K}_{3}\left(  G|s_{0},\left(  \sigma^{1}%
,\widehat{\sigma}^{2},\sigma^{3}\right)  \right)  =0$ or $\mathbf{K}%
_{3}\left(  G|s_{0},\left(  \sigma^{1},\widehat{\sigma}^{2},\sigma^{3}\right)
\right)  =1$, because then we would also have $\widetilde{Q}^{-2}\left(
s_{0},\sigma^{1},\widehat{\sigma}^{2},\sigma^{3}\right)  \geq0$ which
contradicts (\ref{eqNE113}). Hence $\mathbf{K}_{3}\left(  G|s_{0},\left(
\sigma^{1},\widehat{\sigma}^{2},\sigma^{3}\right)  \right)  =2$ and then
\[
Q^{3}\left(  s_{0},\sigma^{1},\widehat{\sigma}^{2},\sigma^{3}\right)
=\widetilde{Q}^{-2}\left(  s_{0},\sigma^{1},\widehat{\sigma}^{2},\sigma
^{3}\right)  ;
\]
hence from (\ref{eqNE113}) we get
\begin{equation}
\forall\sigma^{1},\sigma^{3}:Q^{3}\left(  s_{0},\widehat{\sigma}^{1}%
,\widehat{\sigma}^{2},\widehat{\sigma}^{3}\right)  \geq Q^{3}\left(
s_{0},\sigma^{1},\widehat{\sigma}^{2},\sigma^{3}\right) \nonumber
\end{equation}
and then%
\begin{equation}
\forall\sigma^{3}:Q^{3}\left(  s_{0},\widehat{\sigma}^{1},\widehat{\sigma}%
^{2},\widehat{\sigma}^{3}\right)  =Q^{3}\left(  s_{0},\widehat{\sigma}%
^{1},\widehat{\sigma}^{2},\sigma^{3}\right)  . \label{eqNE114a}%
\end{equation}

\end{enumerate}

\noindent Combining (\ref{eqNE111})-(\ref{eqNE114a}),we see that
\begin{equation}
\mathbf{K}_{3}\left(  G|s_{0},\left(  \widehat{\sigma}^{1},\widehat{\sigma
}^{2},\widehat{\sigma}^{3}\right)  \right)  =2\Rightarrow\forall
n,\forall\sigma^{n}:Q^{n}\left(  s_{0},\widehat{\sigma}\right)  \geq
Q^{n}\left(  s_{0},\sigma^{n},\widehat{\sigma}^{-n}\right)  . \label{eqNE142}%
\end{equation}

\noindent\underline{\textbf{III.} $\mathbf{K}_{3}\left(  G|s_{0},\left(
\widehat{\sigma}^{1},\widehat{\sigma}^{2},\widehat{\sigma}^{3}\right)
\right)  =0$ }.

\begin{enumerate}
\item For all $\sigma^{1},\sigma^{3}$ we have%
\begin{equation}
Q^{1}\left(  s_{0},\widehat{\sigma}^{1},\widehat{\sigma}^{2},\widehat{\sigma
}^{3}\right)  =0=\widetilde{Q}^{-2}\left(  s_{0},\widehat{\sigma}^{1}%
,\widehat{\sigma}^{2},\widehat{\sigma}^{3}\right)  \geq\widetilde{Q}%
^{-2}\left(  s_{0},\sigma^{1},\widehat{\sigma}^{2},\sigma^{3}\right)  .
\label{eqNE131}%
\end{equation}
We cannot have $\mathbf{K}_{3}\left(  G|s_{0},\left(  \sigma^{1}%
,\widehat{\sigma}^{2},\sigma^{3}\right)  \right)  =1$, because then we would
also have $\widetilde{Q}^{-2}\left(  s_{0},\sigma^{1},\widehat{\sigma}%
^{2},\sigma^{3}\right)  >0$, which would contradict (\ref{eqNE131}). If
$\mathbf{K}_{3}\left(  G|s_{0},\left(  \sigma^{1},\widehat{\sigma}^{2}%
,\sigma^{3}\right)  \right)  =2$ or $\mathbf{K}_{3}\left(  G|s_{0},\left(
\sigma^{1},\widehat{\sigma}^{2},\sigma^{3}\right)  \right)  =0$ then
$Q^{1}\left(  s_{0},\sigma^{1},\widehat{\sigma}^{2},\sigma^{3}\right)
=0\ $and so%
\begin{equation}
\forall\sigma^{1},\sigma^{3}:0=Q^{1}\left(  s_{0},\widehat{\sigma}%
^{1},\widehat{\sigma}^{2},\widehat{\sigma}^{3}\right)  =Q^{1}\left(
s_{0},\sigma^{1},\widehat{\sigma}^{2},\sigma^{3}\right)  \Rightarrow
\forall\sigma^{1}:Q^{1}\left(  s_{0},\widehat{\sigma}^{1},\widehat{\sigma}%
^{2},\widehat{\sigma}^{3}\right)  =Q^{1}\left(  s_{0},\sigma^{1}%
,\widehat{\sigma}^{2},\widehat{\sigma}^{3}\right)  . \label{eqNE133}%
\end{equation}

\item By the same argument as in the previous cases we get%
\begin{equation}
\forall\sigma^{2}:Q^{2}\left(  s_{0},\widehat{\sigma}^{1},\widehat{\sigma}%
^{2},\widehat{\sigma}^{3}\right)  \geq Q^{2}\left(  s_{0},\widehat{\sigma}%
^{1},\sigma^{2},\widehat{\sigma}^{3}\right)  . \label{eqNE132a}%
\end{equation}

\item Finally, we have seen that, for all $\left(  \sigma^{1},\sigma
^{3}\right)  $, either $\mathbf{K}_{3}\left(  G|s_{0},\left(  \sigma
^{1},\widehat{\sigma}^{2},\sigma^{3}\right)  \right)  =2$ or $\mathbf{K}%
_{3}\left(  G|s_{0},\left(  \sigma^{1},\widehat{\sigma}^{2},\sigma^{3}\right)
\right)  =0$; in both cases
\[
\forall\sigma^{1},\sigma^{3}:Q^{3}\left(  s_{0},\widehat{\sigma}^{1}%
,\widehat{\sigma}^{2},\widehat{\sigma}^{3}\right)  =\widetilde{Q}^{-2}\left(
s_{0},\widehat{\sigma}^{1},\widehat{\sigma}^{2},\widehat{\sigma}^{3}\right)
\geq\widetilde{Q}^{-2}\left(  s_{0},\sigma^{1},\widehat{\sigma}^{2},\sigma
^{3}\right)  =Q^{3}\left(  s_{0},\sigma^{1},\widehat{\sigma}^{2},\sigma
^{3}\right)
\]
and so%
\begin{equation}
\forall\sigma^{3}:Q^{3}\left(  s_{0},\widehat{\sigma}^{1},\widehat{\sigma}%
^{2},\widehat{\sigma}^{3}\right)  \geq Q^{3}\left(  s_{0},\widehat{\sigma}%
^{1},\widehat{\sigma}^{2},\widehat{\sigma}^{3}\right)  . \label{eqNE132}%
\end{equation}

\end{enumerate}

Combining (\ref{eqNE133})-(\ref{eqNE132}), we see that
\begin{equation}
\mathbf{K}_{3}\left(  G|s_{0},\left(  \widehat{\sigma}^{1},\widehat{\sigma
}^{2},\widehat{\sigma}^{3}\right)  \right)  =0\Rightarrow\forall
n,\forall\sigma^{n}:Q^{n}\left(  s_{0},\widehat{\sigma}\right)  \geq
Q^{n}\left(  s_{0},\sigma^{n},\widehat{\sigma}^{-n}\right)  . \label{eqNE143}%
\end{equation}

\noindent In conclusion, combining (\ref{eqNE141}), (\ref{eqNE142}) and
(\ref{eqNE143}) we see that:\ every profile $\left(  \widehat{\sigma}%
^{1},\widehat{\sigma}^{2},\widehat{\sigma}^{3}\right)  $ which is optimal in
$\widetilde{\Gamma}_{3}^{2}\left(  G|s_{0}\right)  $ is also a NE\ of
$\Gamma_{3}\left(  G|s_{0}\right)  $.
\end{proof}

Before we prove additional facts about $\widetilde{\Gamma}_{3}^{2}\left(
G|s_{0}\right)  $ we need the following.

\begin{definition}
\label{prop0508}A graph $G$ is called \emph{median} if for every three
vertices $x$, $y$, and $z$ there exists a \emph{unique} vertex $m\left(
x,y,z\right)  $ (the \emph{median} vertex of $x,y,z$) which belongs to
shortest paths between each pair of $x,y,z$.
\end{definition}

The following facts are well known \cite{TreeIsMedian}. First, every tree is a
median graph. Second, in a tree the union of the three (unique)\ shortest
paths between the pairs of vertices $x$, $y$, and $z$ is

\begin{enumerate}
\item either a path, in which case the median $m\left(  x,y,z\right)  $ is
equal to one of $x$, $y$, or $z$;

\item or a subtree formed by three paths meeting at a single central
node,which is the median of $x$, $y$, and $z$.
\end{enumerate}

Now we can prove some additional properties of $\widetilde{\Gamma}_{3}%
^{2}\left(  G|s_{0}\right)  $.

\begin{theorem}
\label{prop0509}If $G$ is a path then, for any $s_{0}$, every strategy profile
$\widehat{\sigma}=\left(  \widehat{\sigma}^{1},\widehat{\sigma}^{2}%
,\widehat{\sigma}^{3}\right)  $ which is optimal in $\widetilde{\Gamma}%
_{3}^{2}\left(  G|s_{0}\right)  $ is capturing.
\end{theorem}

\begin{proof}
If $s_{0}$ is a capture state, the theorem is obviously true. Take any
non-capture starting state $s_{0}=\left(  x_{0}^{1},x_{0}^{2},x_{0}%
^{3},p\right)  $; since $G$ is a path, it is a median graph and one of
$x_{0}^{1},x_{0}^{2},x_{0}^{3}$ is the median of the other two (we will also
say that either $x_{0}^{1},x_{0}^{2},x_{0}^{3}$\ or $P_{1},P_{2},P_{3}$ are
\emph{collinear}). We define strategies $\overline{\sigma}^{1}$,
$\overline{\sigma}^{2}$, $\overline{\sigma}^{3}$ for each case.

\begin{enumerate}
\item If $x_{0}^{1}$ is the median of $x_{0}^{2}$ and $x_{0}^{3}$, then:
$P_{1}$ moves towards $P_{2}$, $P_{2}$ moves away from $P_{1}$ and $P_{3}$
stays in place; eventually $P_{2}$ is captured.

\item If $x_{0}^{2}$ is the median of $x_{0}^{1}$ and $x_{0}^{3}$, then:
$P_{1}$ moves towards $P_{2}$, $P_{2}$ moves towards $P_{3}$ and $P_{3}$ moves
away from $P_{2}$; eventually $P_{3}$ is captured.

\item If $x_{0}^{3}$ is the median of $x_{0}^{1}$ and $x_{0}^{2}$, then:
$P_{1}$ moves towards $P_{2}$, $P_{2}$ moves away from $P_{1}$ and $P_{3}$
moves away from $P_{2}$; eventually $P_{2}$ is captured.
\end{enumerate}

\noindent\noindent In every case $\overline{\sigma}=\left(  \overline{\sigma
}^{1},\overline{\sigma}^{2},\overline{\sigma}^{3}\right)  $ is capturing
\emph{and} optimal. Since $\overline{\sigma}$ is capturing, the same holds for
\emph{every} optimal profile $\widehat{\sigma}$, because they all yield the
same payoff.
\end{proof}

The above defined $\overline{\sigma}$ will be called \emph{path strategies}
and will be used to prove the following lemma, needed to extend Theorem
\ref{prop0509} to trees.

\begin{lemma}
\label{prop0510}If $G$ is a tree then there exists a positional profile
$\widetilde{\sigma}=\left(  \widetilde{\sigma}^{1},\widetilde{\sigma}%
^{2},\widetilde{\sigma}^{3}\right)  $ for which the following hold in
$\widetilde{\Gamma}_{3}^{2}\left(  G|s_{0}\right)  $. \ 

\begin{enumerate}
\item For every $s_{0}$: $\left(  s_{0},\widetilde{\sigma}^{1},\widetilde
{\sigma}^{2},\widetilde{\sigma}^{3}\right)  $ results in capture.

\item If $\left(  s_{0},\widetilde{\sigma}^{1},\widetilde{\sigma}%
^{2},\widetilde{\sigma}^{3}\right)  $ results in capture \emph{of} $P_{2}$
then, for every $\sigma^{2}$, $\left(  s_{0},\widetilde{\sigma}^{1},\sigma
^{2},\widetilde{\sigma}^{3}\right)  $ results in capture \emph{of} $P_{2}$.

\item If $\left(  s_{0},\widetilde{\sigma}^{1},\widetilde{\sigma}%
^{2},\widetilde{\sigma}^{3}\right)  $ results in capture \emph{by} $P_{2}$
then for every $\sigma^{1},\sigma^{3}$, $\left(  s_{0},\sigma^{1}%
,\widetilde{\sigma}^{2},\sigma^{3}\right)  $ results in capture \emph{by}
$P_{2}$.
\end{enumerate}
\end{lemma}

\begin{proof}
A rough description of $\widetilde{\sigma}\ $is quite simple:\ each player
tries to reach the median as fast as possible; as soon as this happens the
players are collinear and they start playing their path strategies. We next
give a (straightforward but rather tedious)\ rigorous proof. In what follows,
we denote the median of $x_{t}^{1},x_{t}^{2},x_{t}^{3}\ $ by $m_{t}$.

\bigskip

\noindent For Part 1 of the theorem, we distinguish two cases.

\noindent\underline{\textbf{Case A}}\textbf{.} Suppose that at some time $t$
the game state is $\left(  x_{t}^{1},x_{t}^{2},x_{t}^{3},p\right)  $ where one
of $x_{t}^{1},x_{t}^{2},x_{t}^{3}$ is the median of the other two (they are
collinear). In this case $\widetilde{\sigma}=\overline{\sigma}$, i.e., the
players use the path strategies of Theorem \ref{prop0509}. It is easily
checked that the players remain collinear for the rest of the game and a
capture results.

\noindent\underline{\textbf{Case B}}\textbf{.} Suppose that $x_{0}^{1}%
,x_{0}^{2},x_{0}^{3}$ are not collinear. The initial part of $\widetilde
{\sigma}^{1},\widetilde{\sigma}^{2},\widetilde{\sigma}^{3}$ prescribes that
every player moves directly towards the median $m_{t}$. As a result, let
$t_{0}$ denote the first time when a (single) player $P_{n}$ is at distance 1
from $m_{t}$, as depicted in Figures \ref{fig0503a} and \ref{fig0504} (in the
figures we only show the subtree of $G$ which is defined by the positions of
$P_{1}$, $P_{2}$ and $P_{3}$; dotted lines indicate paths of length one or
more). We will now define $\widetilde{\sigma}^{1},\widetilde{\sigma}%
^{2},\widetilde{\sigma}^{3}$ depending on which $P_{n}$ first reaches $m_{t}$;
when some move is not specified, the respective strategy can be defined arbitrarily.

\begin{enumerate}
\item Suppose $P_{n}=P_{1}$, i.e. at $t_{0}$ we have $d\left(  x_{t_{0}}%
^{1},m_{t_{0}}\right)  =1$, as shown in Fig. \ref{fig0503a}.a. $P_{3}$ stays
in place at $t_{0}+2$, $P_{1}$ enters $m_{t}$ at $t_{0}+3$ and the players
become collinear. Now every player starts using his path strategy. It is easy
to check that this results in capture of $P_{2}$.\begin{figure}[ptb]
\begin{center}
\begin{tikzpicture}
\tikzstyle{every node}=[draw,shape=circle];
\node[label=below:$P_1$](v1) at (0,0){};
\node(v1) at (0,0){};
\node[label=above:$m_{t_0}$](v2)   at (1,0){};
\node[label=below:$P_3$](v3) at (1,-1){};
\node[label=below:$P_2$](v4) at (2,0){};
\draw[line width=0.5mm](v1) -- (v2);
\draw[line width=0.5mm,dotted](v2) -- (v3);
\draw[line width=0.5mm,dotted](v2) -- (v4);
\node[draw=none,fill=none] at (-0.7,-1.7) {\textbf{(a)}};
\draw (-1,1) -- (3,1) -- (3,-2) -- (-1,-2) -- (-1,1);
\end{tikzpicture}
\begin{tikzpicture}
\tikzstyle{every node}=[draw,shape=circle];
\node[label=below:$P_1$](v1) 	at (-0.6,0){};
\node[label=above:$a$](v1a) at (0.2,0){};
\node[label=above:$m_{t_0}$](v2)   	at (1,0){};
\node[label=below:$P_3$](v3) 	at (1,-1){};
\node[label=below:$P_2$](v4) 	at (2,0){};
\node(v4) 	at (2,0){};
\draw[line width=0.5mm](v1) -- (v1a);
\draw[line width=0.5mm](v1a) -- (v2);
\draw[line width=0.5mm,dotted](v2) -- (v3);
\draw[line width=0.5mm](v2) -- (v4);
\node[draw=none,fill=none] at (-0.7,-1.7) {\textbf{(b)}};
\draw (-1,1) -- (3,1) -- (3,-2) -- (-1,-2) -- (-1,1);
\end{tikzpicture}
\begin{tikzpicture}
\tikzstyle{every node}=[draw,shape=circle];
\node[label=below:$P_1$](v1) at (-0.6,0){};
\node[label=above:$a$](v1a) at (0.2,0){};
\node[label=above:$b$](v1b) at (1.0,0){};
\node[label=above:$m_{t_0}$](v2)   at (1.8,0){};
\node[label=below:$P_3$](v3) at (1.8,-1){};
\node[label=below:$P_2$](v4) at (2.6,0){};
\draw[line width=0.5mm,dotted](v1) -- (v1a);
\draw[line width=0.5mm](v1a) -- (v1b);
\draw[line width=0.5mm](v1b) -- (v2);
\draw[line width=0.5mm,dotted](v2) -- (v3);
\draw[line width=0.5mm](v2) -- (v4);
\node[draw=none,fill=none] at (-0.7,-1.7) {\textbf{(c)}};
\draw (-1,1) -- (3,1) -- (3,-2) -- (-1,-2) -- (-1,1);
\end{tikzpicture}
\end{center}
\caption{Possible placements of $P_{1},P_{2},P_{3}$ in the proof of Theorem
\ref{prop0510}.}%
\label{fig0503a}%
\end{figure}

\item Suppose $P_{n}=P_{2}$, i.e., at $t_{0}$ we have $d\left(  x_{t_{0}}%
^{2},m_{t_{0}}\right)  =1$. Now we must distinguish two sub-cases.

\begin{enumerate}
\item Say $d\left(  x_{t_{0}}^{1},m_{t_{0}}\right)  =2$, as in Fig.
\ref{fig0503a}.b. $P_{3}$ stays in place at $t_{0}+1$ and $t_{0}+4$, $P_{1}$
moves to $a$ at $t_{0}+2$ and to\ $m_{t}$ at $t_{0}+5$, when the players
become collinear and start playing their path strategies; eventually $P_{2}$
is captured.

\item Say $d\left(  x_{t_{0}}^{1},m_{t_{0}}\right)  \geq3$, as in Fig.
\ref{fig0503a}.c. $\ P_{2}$ enters $m_{t}$ at time $t_{0}+3$, the players
become collinear and start playing their path strategies; eventually $P_{3}$
is captured.
\end{enumerate}

\item Suppose $P_{n}=P_{3}$, i.e. at $t_{0}$ we have $d\left(  x_{t_{0}}%
^{3},m_{t_{0}}\right)  =1$. Now we must distinguish four sub-cases.

\begin{figure}[ptbh]
\begin{center}
\begin{tikzpicture}
\tikzstyle{every node}=[draw,shape=circle];
\node(v1) at (0.0,0){};
\node[draw=none](w1) at (0.0,0.4){$P_1$};
\node(v2) at (0.6,0){};
\node[draw=none](w2) at (0.6,0.4){$a$};
\node(v3)  at (1.2,0){};
\node[draw=none](w3)   at (1.2,0.4){$m_{t_0}$};
\node(v4)  at (1.8,0.0){};
\node[draw=none](w4) at (1.8,0.4){$u$};
\node(v5) at (2.4,0){};
\node[draw=none](w5) at (2.4,0.4){$P_2$};
\node(v6)  at (1.2,-1.0){};
\node[draw=none](w6) at (1.2,-1.5){$P_3$};
\draw[line width=0.5mm](v1) -- (v2);
\draw[line width=0.5mm](v2) -- (v3);
\draw[line width=0.5mm](v3) -- (v4);
\draw[line width=0.5mm](v4) -- (v5);
\draw[line width=0.5mm](v3) -- (v6);
\node[draw=none,fill=none] at (0.1,-1.8) {\textbf{(a)}};
\draw (-0.3,1.0) -- (2.6,1.0) -- (2.6,-2.0) -- (-0.2,-2.0) -- (-0.3,1.0);
\end{tikzpicture}
\begin{tikzpicture}
\tikzstyle{every node}=[draw,shape=circle];
\node(v1) at (0.0, 0.0){};
\node[draw=none](w1) at (0.0, 0.5){$P_1$};
\node(v2) at (0.6, 0.0){};
\node[draw=none](w2) at (0.6, 0.5){$a$};
\node(v3) at (1.2, 0.0){};
\node[draw=none](w3) at (1.2, 0.5){$b$};
\node(v4) at (1.8, 0.0){};
\node[draw=none](w4) at (1.8, 0.5){$m_{t_0}$};
\node(v5) at (2.4, 0.0){};
\node[draw=none](w5) at (2.4, 0.5){$u$};
\node(v6) at (3.0, 0.0){};
\node[draw=none](w6) at (3.0, 0.5){$P_2$};
\node(v7) at (1.8,-1.0){};
\node[draw=none](w7a) at (1.8,-1.5){$P_3,$};
\draw[line width=0.5mm](v1) -- (v2);
\draw[line width=0.5mm](v2) -- (v3);
\draw[line width=0.5mm](v3) -- (v4);
\draw[line width=0.5mm](v4) -- (v5);
\draw[line width=0.5mm](v5) -- (v6);
\draw[line width=0.5mm](v4) -- (v7);
\node[draw=none,fill=none] at (0.1,-1.8) {\textbf{(b)}};
\draw (-0.3,1.0) -- (3.2,1.0) -- (3.2,-2.0) -- (-0.3,-2.0) -- (-0.3,1.0);
\end{tikzpicture}
\begin{tikzpicture}
\tikzstyle{every node}=[draw,shape=circle];
\node(v1) at (0.0, 0.6){};
\node[draw=none](w1) at (0.0, 1.1){$P_1$};
\node(v2) at (0.6, 0.6){};
\node[draw=none](w2) at (0.6, 1.1){$a$};
\node(v3) at (1.2, 0.6){};
\node[draw=none](w3) at (1.2, 1.1){$b$};
\node(v4) at (1.8, 0.6){};
\node[draw=none](w4) at (1.8, 1.1){$c$};
\node(v5) at (2.4, 0.6){};
\node[draw=none](w5) at (2.4, 1.1){$m_{t_0}$};
\node(v6) at (3.0, 0.6){};
\node[draw=none](w6) at (3.0, 1.1){$u$};
\node(v7) at (3.6, 0.6){};
\node[draw=none](w7) at (3.6, 1.1){$P_2$};
\node(v8) at (2.4,-0.4){};
\node[draw=none](w8a) at (2.4,-1.1){$P_3$};
\draw[line width=0.5mm,dotted](v1) -- (v2);
\draw[line width=0.5mm](v2) -- (v3);
\draw[line width=0.5mm](v3) -- (v4);
\draw[line width=0.5mm](v4) -- (v5);
\draw[line width=0.5mm](v5) -- (v6);
\draw[line width=0.5mm](v6) -- (v7);
\draw[line width=0.5mm](v5) -- (v8);
\node[draw=none,fill=none] at (0.0,-1.3) {\textbf{(c)}};
\draw (-0.3,1.45) -- (4.0,1.45) -- (4.0,-1.50) -- (-0.3,-1.50) -- (-0.3,1.45);
\end{tikzpicture}
\begin{tikzpicture}
\tikzstyle{every node}=[draw,shape=circle];
\node(v1) at(-0.2,0.6){};
\node[draw=none](w1) at (-0.2,1.1){$P_1$};
\node(v2) at (0.6,0.6){};
\node[draw=none](w2) at ( 0.6,1.1){$m_{t_0}$};
\node(v3)  at (1.2,0.6){};
\node[draw=none](w2)   at (1.2,1.1){$v$};
\node(v4)  at (1.8, 0.6){};
\node[draw=none](w3) at (1.8,1.1){$u$};
\node(v5) at (2.8,0.6){};
\node[draw=none](w5a) at (2.8,1.1){$P_2,$};
\node(v6)  at (0.6,-0.4){};
\node[draw=none](w6a) at (0.6,-1.1){$P_3,$};
\draw[line width=0.5mm,dotted](v1) -- (v2);
\draw[line width=0.5mm](v2) -- (v3);
\draw[line width=0.5mm](v3) -- (v4);
\draw[line width=0.5mm,dotted](v4) -- (v5);
\draw[line width=0.5mm](v2) -- (v6);
\node[draw=none,fill=none] at (-0.2,-1.3) {\textbf{(d)}};
\draw (-0.5,1.50) -- (3.4,1.50) -- (3.4,-1.45) -- (-0.5,-1.45) -- (-0.5,1.50);
\end{tikzpicture}
\end{center}
\caption{Possible placements of $P_{1},P_{2},P_{3}$ in the proof of Theorem
\ref{prop0510}.}%
\label{fig0504}%
\end{figure}

\begin{enumerate}
\item Say $d\left(  x_{t}^{1},m_{t}\right)  =2$, $d\left(  x_{t}^{2}%
,m_{t}\right)  =2$, as in Fig. \ref{fig0504}.a. $P_{1}$ enters $a$ at
$t_{0}+1$, $P_{2}$ enters $u$ at $t_{0}+2$ and $P_{3}$ stays in place at
$t_{0}+3$. At $t_{0}+4$ $P_{1}$ enters $m_{t}$, the players become collinear
and start playing their path strategies; eventually $P_{2}$ is captured.

\item Say $d\left(  x_{t}^{1},m_{t}\right)  =3$, $d\left(  x_{t}^{2}%
,m_{t}\right)  =2$, as in Fig. \ref{fig0504}.b. $P_{1}$ enters $a$ at
$t_{0}+1$, $P_{2}$ enters $u$ at $t_{0}+2$ and $P_{3}$ stays in place at
$t_{0}+3$. At $t_{0}+4$ $P_{1}$ enters $b$ and at $t_{0}+7$ he enters $m_{t}$,
the players become collinear and start playing their path strategies;
eventually $P_{2}$ is captured.

\item Say $d\left(  x_{t}^{1},m_{t}\right)  \geq4$, $d\left(  x_{t}^{2}%
,m_{t}\right)  =2$, as in Fig. \ref{fig0504}.c. $P_{2}$ enters $u$ at
$t_{0}+2$ and $m_{t}$ at $t_{0}+5$, at $\ $which time the players become
collinear and start playing their path strategies; eventually $P_{3}$ is captured.

\item Say $d\left(  x_{t}^{2},m_{t}\right)  \geq3$, as in Fig. \ref{fig0504}%
.d. At $t_{0}+3$ $P_{3}$ enters $m_{t}$ , the players become collinear and
start playing their path strategies; eventually $P_{2}$ is
captured.\FloatBarrier

\end{enumerate}
\end{enumerate}

\noindent\noindent This completes the description of $\widetilde{\sigma
}=\left(  \widetilde{\sigma}^{1},\widetilde{\sigma}^{2},\widetilde{\sigma}%
^{3}\right)  $ and it is readily seen that it always leads to capture; so the
first part of the theorem has been proved.

\bigskip

\noindent For Part 2 of the theorem, assume that $\left(  s_{0},\widetilde
{\sigma}^{1},\widetilde{\sigma}^{2},\widetilde{\sigma}^{3}\right)  $ leads to
$P_{2}$ capture; there are two ways for this to happen. Either the players are
collinear in $s_{0}$ and $P_{2}$ is \emph{not} in the middle; in this case
$P_{2}$ is captured for every $\sigma^{2}$ he uses. Or the players are not
collinear in $s_{0}$ but eventually reach one of cases 1, 2.a, 3.a, 3.b, 3.d;
in this case, if $P_{2}$ uses a $\sigma^{2}$ which deviates from
$\widetilde{\sigma}^{2}$, he will approach $m_{t}$ no faster than if he used
$\widetilde{\sigma}^{2}$ and a straightforward examination of cases 1, 2.a,
3.a, 3.b, 3.d\ shows that $\left(  s_{0},\widetilde{\sigma}^{1},\sigma
^{2},\widetilde{\sigma}^{3}\right)  $ will also lead to capture of $P_{2}$.

\bigskip

\noindent The proof of Part 3 is similar to that of Part 2 and hence omitted.
\end{proof}

Now we can expand Theorem \ref{prop0509} from paths to trees.

\begin{theorem}
\label{prop0511}If $G$ is a tree then, for any $s_{0}$, every strategy profile
$\widehat{\sigma}=\left(  \widehat{\sigma}^{1},\widehat{\sigma}^{2}%
,\widehat{\sigma}^{3}\right)  $ which is optimal in $\widetilde{\Gamma}%
_{3}^{2}\left(  G|s_{0}\right)  $ is capturing.
\end{theorem}

\begin{proof}
Let $\widehat{\sigma}=\left(  \widehat{\sigma}^{1},\widehat{\sigma}%
^{2},\widehat{\sigma}^{3}\right)  $ be an optimal (for any $s_{0}$)\ strategy
profile and take some $s_{0}$ such that $\left(  s_{0},\widetilde{\sigma}%
^{1},\widetilde{\sigma}^{2},\widetilde{\sigma}^{3}\right)  $ leads to capture
of $P_{2}$. Then, by Part 2 of Lemma \ref{prop0510}, we have%
\[
\forall\sigma^{2}:\widetilde{Q}^{2}\left(  s_{0},\widetilde{\sigma}^{1}%
,\sigma^{2},\widetilde{\sigma}^{3}\right)  <0
\]
and so%
\[
\widetilde{Q}^{2}\left(  s_{0},\widehat{\sigma}^{1},\widehat{\sigma}%
^{2},\widehat{\sigma}^{3}\right)  =\min_{\sigma^{1},\sigma^{3}}\max
_{\sigma^{2}}\widetilde{Q}^{2}\left(  s_{0},\sigma^{1},\sigma^{2},\sigma
^{3}\right)  \leq\max_{\sigma^{2}}\widetilde{Q}^{2}\left(  s_{0}%
,\widetilde{\sigma}^{1},\sigma^{2},\widetilde{\sigma}^{3}\right)  <0;
\]
in other words, $\left(  s_{0},\widehat{\sigma}^{1},\widehat{\sigma}%
^{2},\widehat{\sigma}^{3}\right)  $ leads to capture of $P_{2}$. Similarly, by
Part 3 of Lemma \ref{prop0510} we can prove that, for every $s_{0}$ such that
$\left(  s_{0},\widetilde{\sigma}^{1},\widetilde{\sigma}^{2},\widetilde
{\sigma}^{3}\right)  $ leads to capture of $P_{3}$, the same holds for
$\left(  s_{0},\widehat{\sigma}^{1},\widehat{\sigma}^{2},\widehat{\sigma}%
^{3}\right)  $. Since, by Part 1 of Lemma \ref{prop0510}, for every $s_{0}$,
$\left(  s_{0},\widetilde{\sigma}^{1},\widetilde{\sigma}^{2},\widetilde
{\sigma}^{3}\right)  $ leads to capture (of either $P_{2}$ or $P_{3}$) we have
proved the theorem.
\end{proof}

Now we return to the three-player game $\Gamma_{3}\left(  G|s_{0}\right)  $
and show that: if $G$ is a tree, then $\Gamma_{3}\left(  G|s_{0}\right)  $ has
a capturing NE for every initial state $s_{0}$ (hence, while the converse of
Theorem \ref{prop0505} does not hold for every cop-win graph, it holds for the
special case of trees).

\begin{theorem}
\label{prop0512}If $G$ is a tree, then%
\[
\forall s_{0}\text{ there exists a NE }\widehat{\sigma}\text{ of }\Gamma
_{3}\left(  G|s_{0}\right)  :\mathbf{K}_{3}\left(  G|s_{0},\widehat{\sigma
}\right)  >0.
\]

\end{theorem}

\begin{proof}
Since $G$ is a tree, every optimal profile $\widehat{\sigma}=\left(
\widehat{\sigma}^{1},\widehat{\sigma}^{2},\widehat{\sigma}^{3}\right)  $ of
$\widetilde{\Gamma}_{3}^{2}\left(  G|s_{0}\right)  $ is capturing in both
$\widetilde{\Gamma}_{3}^{2}\left(  G|s_{0}\right)  $ (by Theorem
\ref{prop0511}); and in $\Gamma_{3}\left(  G|s_{0}\right)  $ (since the two
games are played by the same rules). Hence $\mathbf{K}_{3}\left(
G|s_{0},\widehat{\sigma}\right)  >0$. And $\widehat{\sigma}\ $is a NE of
$\Gamma_{3}\left(  G|s_{0}\right)  $ by Theorem \ref{prop0507} (since every
tree $G$ has $c\left(  G\right)  =1$).
\end{proof}

We conclude this section with a result for graphs which are not cop-win.

\begin{theorem}
\label{prop0513}$c\left(  G\right)  >1\Rightarrow\left(  \exists s_{0}%
:\Gamma_{3}\left(  G|s_{0}\right)  \text{ has a noncapturing NE }%
\widehat{\sigma}\right)  $
\end{theorem}

\begin{proof}
We will construct the required $s_{0}$ and $\widehat{\sigma}=\left(
\widehat{\sigma}^{1},\widehat{\sigma}^{2},\widehat{\sigma}^{3}\right)  $.
Since $c\left(  G\right)  >1$, there exist an $\widetilde{s}_{0}=\left(
x^{1},x^{2},p\right)  $ and a $\Gamma_{2}\left(  G|\widetilde{s}_{0}\right)
$-optimal noncapturing profile $\widetilde{\sigma}=\left(  \widetilde{\sigma
}^{1},\widetilde{\sigma}^{2}\right)  $. Now let $s_{0}=\left(  x^{1}%
,x^{2},x^{1},1\right)  $ and define the $\Gamma_{3}\left(  G|\widetilde{s}%
_{0}\right)  $ strategies as follows:\ $\widehat{\sigma}^{2}$ is
$\widetilde{\sigma}^{2}$ (expanded to work in $\Gamma_{3}\left(
G|\widetilde{s}_{0}\right)  $ )\ and, for $n\in\left\{  1,3\right\}  $,
$\widehat{\sigma}^{n}$ specifies that $P_{n}$ always stays in place. Then,
since $\widehat{\sigma}^{2}$ is an optimal evasion strategy we have:
\[
\forall\sigma^{1}:Q^{1}\left(  s_{0},\widehat{\sigma}^{1},\widehat{\sigma}%
^{2},\widehat{\sigma}^{3}\right)  =0=Q^{1}\left(  s_{0},\sigma^{1}%
,\widehat{\sigma}^{2},\widehat{\sigma}^{3}\right)  ;\text{ }%
\]
also, since $P_{2}$ must enter $x^{1}$ to capture $P_{3}$, but then he would
first be captured by $P_{1}$, we have
\[
\forall\sigma^{2}:Q^{2}\left(  s_{0},\widehat{\sigma}^{1},\widehat{\sigma}%
^{2},\widehat{\sigma}^{3}\right)  =0\geq Q^{2}\left(  s_{0},\widehat{\sigma
}^{1},\sigma^{2},\widehat{\sigma}^{3}\right)  ;\text{ }%
\]
and since $P_{3}$ never receives positive payoff we have
\[
\text{ }\forall\sigma^{3}:Q^{3}\left(  s_{0},\widehat{\sigma}^{1}%
,\widehat{\sigma}^{2},\widehat{\sigma}^{3}\right)  =0\geq Q^{3}\left(
s_{0},\widehat{\sigma}^{1},\widehat{\sigma}^{2},\sigma^{3}\right)  .\text{ }%
\]
So $s_{0}$ is a noncapturing NE of $\Gamma\left(  G|s_{0}\right)  $.
\end{proof}

\section{GCR\ with $N$ Players, $N\geq4$\label{sec06}}

We will now briefly examine $\Gamma_{N}\left(  G|s_{0}\right)  $ for $N\geq4$.
Most of the game elements have been defined in Section \ref{sec02}; we define
the turn payoffs $q^{n}$ by generalizing (\ref{eq0501a}). Namely, at every
turn $P_{n}$ receives:

\begin{enumerate}
\item a payoff of $-1$ if he is captured by $P_{n}$;

\item a payoff of $1$ if he captures $P_{n+1}$, but is not simultaneously
captured by $P_{n-1}$;

\item a payoff of $0$ in every other case.
\end{enumerate}

\noindent The above turn payoffs and the total payoff $Q^{n}$ of
(\ref{eq02001})\ complete the specification of $\Gamma_{N}\left(
G|s_{0}\right)  $.

Since it is a multi-player discounted stochastic game of perfect information,
$\Gamma_{N}\left(  G|s_{0}\right)  $ \ has (by Theorem \ref{prop0301}) a NE in
deterministic positional strategies. The $N$-player analog of Theorem
\ref{prop0502} also holds.

\begin{theorem}
\label{prop0601}For every $G,s_{0}$ and $\gamma$, $\Gamma_{N}\left(
G|s_{0}\right)  $ has a NE $\widehat{\pi}=(\widehat{\pi}^{1},\widehat{\pi}%
^{2},...,\widehat{\pi}^{N})$ in deterministic (generally non-positional) strategies.
\end{theorem}

\begin{proof}
The proof involves the use of the auxiliary two-player, zero-sum games
$\widetilde{\Gamma}_{N}^{1}\left(  G|s_{0}\right)  ,...,\widetilde{\Gamma}%
_{N}^{N}\left(  G|s_{0}\right)  $ . In $\widetilde{\Gamma}_{N}^{n}\left(
G|s_{0}\right)  $, $P_{n}$ plays against $P_{-n}$, who controls the tokens
$P_{1},...,P_{n-1},P_{n+1},...,P_{N}$. The threat strategies $\widehat{\pi
}=(\widehat{\pi}^{1},\widehat{\pi}^{2},...,\widehat{\pi}^{N})$ are defined in
the same manner as in Section \ref{sec0501}, in terms of the strategies
$\left(  \widehat{\phi}_{n}^{m}\right)  _{m,n\in\left[  N\right]  }$ which are
optimal in the corresponding $\widetilde{\Gamma}_{N}^{n}\left(  G|s_{0}%
\right)  $ games. The rest of the proof is omitted, since it follows closely
that of Theorem \ref{prop0502}.
\end{proof}

Similarly to $\Gamma_{3}\left(  G|s_{0}\right)  $, if $\Gamma_{N}\left(
G|s_{0}\right)  $ has a capturing NE for every initial state $s_{0}$, then $G$
is cop-win. This is stated in the following theorem, where $\mathbf{K}%
_{N}\left(  G|s_{0},\sigma\right)  $ is the obvious generalization of the
capturability function $\mathbf{K}_{3}\left(  G|s_{0},\sigma\right)  $ (the
proof is omitted, since it is similar to that of Theorem \ref{prop0505}).

\begin{theorem}
\label{prop0602}The following holds for every $G$:%
\begin{equation}
\left(  \forall s_{0}\text{ there exists a NE }\widehat{\sigma}\text{ of
}\Gamma_{N}\left(  G|s_{0}\right)  :\text{ }\mathbf{K}_{N}\left(
G|s_{0},\widehat{\sigma}\right)  >0\right)  \Rightarrow c\left(  G\right)
=1.\label{eq0601}%
\end{equation}

\end{theorem}

On the other hand, Theorem \ref{prop0507} does not generalize to the case
$N\geq4$. The following example shows that, even when $G$ is a path, there may
exist optimal profiles $\widehat{\sigma}$ of $\widetilde{\Gamma}_{N}%
^{n}\left(  G|s_{0}\right)  $ which are not NE of $\Gamma_{N}\left(
G|s_{0}\right)  $.

\begin{example}
\label{prop0603}\normalfont In Figure \ref{fig0601} $G$ is a path, the tokens
are positioned as depicted and $P_{4}$ has the starting move; in short
$s_{0}=\left(  1,3,4,5,4\right)  $. 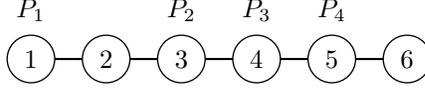
\begin{figure}[ptbh]
\begin{center}
\begin{tikzpicture}
\label{fig0601}
\SetGraphUnit{2}
\Vertex[x=-1,y= 0]{1}
\Vertex[x= 0,y= 0]{2}
\Vertex[x= 1,y= 0]{3}
\Vertex[x= 2,y= 0]{4}
\Vertex[x= 3,y= 0]{5}
\Vertex[x= 4,y= 0]{6}
\node(A) [label=$P_1$] at(-1,0.25) {};
\node(B) [label=$P_2$] at (1,0.25) {};
\node(C) [label=$P_3$] at (2,0.25) {};
\node(D) [label=$P_4$] at (3,0.25) {};
\Edge(1)(2)
\Edge(2)(3)
\Edge(3)(4)
\Edge(4)(5)
\Edge(5)(6)
\SetVertexNoLabel
\end{tikzpicture}
\end{center}
\caption{A path $G$ in which a $\widetilde{\Gamma}_{4}^{2}\left(
G|s_{0}\right)  $-optimal strategy profile is not a NE of $\Gamma_{4}%
^{2}\left(  G|s_{0}\right)  $.}%
\end{figure}In the game $\widetilde{\Gamma}_{4}^{2}\left(  G|s_{0}\right)  $,
$P_{2}$ plays against $P_{-2}$ who controls $P_{1},P_{3},P_{4}$. Clearly the
optimal $P_{-2}$ move from $s_{0}$ is to move $P_{4}$ into vertex $4$, since
then the game ends and $P_{-2}$ receives his maximum possible payoff of 0
(otherwise, on his first move $P_{2}$ captures $P_{3}$ and $P_{-2}$ receives
negative payoff). So every $\widehat{\sigma}^{-2}=\left(  \widehat{\sigma}%
^{1},\widehat{\sigma}^{3},\widehat{\sigma}^{4}\right)  $ which is optimal in
$\widetilde{\Gamma}_{4}^{2}\left(  G|s_{0}\right)  $ must satisfy
$\widehat{\sigma}^{4}\left(  s_{0}\right)  =4$. But such a $\widehat{\sigma
}^{-2}$ cannot be (part of) a NE\ of $\Gamma_{4}\left(  G|s_{0}\right)  $,
because in this game $P_{4}$ can improve his payoff by moving from 5 to 6,
rather than 4.
\end{example}

In Section \ref{sec05} we have shown (Theorem \ref{prop0512})\ that, when $G$
is a tree, for every $s_{0}$ there exists a capturing NE of $\Gamma_{3}\left(
G|s_{0}\right)  $; the proof depended on Theorem \ref{prop0507} which, as
seen, does not generalize for $N\geq4$. Hence we have not been able to
generalize Theorem \ref{prop0512} either. On the other hand, we have not found
a counterexample (i.e., a tree and some initial state for which no capturing
NE exists) hence the matter remains open.

The following generalizes Theorem \ref{prop0513} and is proved very similarly.

\begin{theorem}
\label{prop0604a}For every $N\geq3$ we have
\[
c\left(  G\right)  >1\Rightarrow\left(  \exists s_{0}:\Gamma_{N}\left(
G|s_{0}\right)  \text{ has a noncapturing NE }\widehat{\sigma}\right)  .
\]

\end{theorem}

\section{More Multi-player Pursuit Games\label{sec07}}

In Section \ref{sec02} we have developed a framework which we have used in
Sections \ref{sec04}, \ref{sec05} and \ref{sec06} to study the game
$\Gamma_{N}\left(  G|s_{0}\right)  $, for various $N$ values. As we will now
explain, this framework applies to a wider family of graph pursuit games.

We have in mind games played by players $P_{1}$, $P_{2}$, ..., $P_{N}$ who
take turns in moving tokens along the edges of a graph. For the time being
assume that each player controls one token and has, in general, two goals:
(i)\ to capture some (other players')\ tokens and (ii)\ to avoid capture of
his own token.

Any such situation can be described, by the formulation of Section
\ref{sec02}, as a multi-player discounted stochastic game of perfect
information. Assuming, without loss of generality, that the players move in
the sequence implied by their numbering, the actual \textquotedblleft capture
relationship\textquotedblright\ will be encoded by the turn payoff functions
$q^{n}$. To preserve the semantics of pursuit / evasion, they should have the
form\footnote{The conditions in (\ref{eq0701}) encode ``minimum'' requirements,
additional restrictions may be imposed, e.g., no simultaneous captures.}
\begin{equation}
q^{n}\left(  \left(  x^{1},...,x^{N},p\right)  \right)  =\left\{
\begin{array}
[c]{rll}%
1 & \text{when for some }m: & x^{n}=x^{m}\text{, }m\in A^{n},\\
-1 & \text{when for some }m: & x^{n}=x^{m}\text{, }m\in B^{n},\\
0 & \text{else.} &
\end{array}
\right.  \label{eq0701}%
\end{equation}
where

\begin{enumerate}
\item $A^{n}$ is the set of $P_{n}$'s \textquotedblleft
targets\textquotedblright\ (i.e., the players whom he can capture)\ and

\item $B^{n}$ is the set of $P_{n}$'s \textquotedblleft
pursuers\textquotedblright\ (i.e., the players who can capture him).
\end{enumerate}

\noindent For example, in in $\Gamma_{2}\left(  G|s_{0}\right)  $ we have
players $P_{1}$ and $P_{2}$ with respective sets%
\[
A^{1}=\left\{  P^{2}\right\}  ,B^{1}=\emptyset\text{,\qquad}A^{2}%
=\emptyset,B^{2}=\left\{  P_{1}\right\}  \text{;}%
\]
while in $\Gamma_{3}\left(  G|s_{0}\right)  $ we have players $P_{1},P_{2}$
and $P_{3}$ with respective sets%
\[
A^{1}=\left\{  P^{2}\right\}  ,B^{1}=\emptyset\text{,\qquad}A^{2}=\left\{
P^{3}\right\}  ,B^{2}=\left\{  P_{1}\right\}  \text{,}\qquad A^{3}%
=\emptyset,B^{3}=\left\{  P_{2}\right\}
\]
and the additional condition of no simultaneous captures (which, as,
requires a small modification of (\ref{eq0701})).

As a final example, consider a game which we could call \textquotedblleft%
\emph{Cyclic Cops and Robbers}\textquotedblright; it involves players
$P_{1},P_{2}$ and $P_{3}$ in which: $P_{1}$ chases $P_{2}$ and avoids $P_{3}$;
$P_{2}$ chases $P_{3}$ and avoids $P_{1}$; $P_{3}$ chases $P_{1}$ and avoids
$P_{2}$. In this game we will have
\[
A^{1}=\left\{  P^{2}\right\}  ,B^{1}=\left\{  P^{3}\right\}  \text{,\qquad
}A^{2}=\left\{  P^{3}\right\}  ,B^{2}=\left\{  P_{1}\right\}  \text{,}\qquad
A^{3}=\left\{  P_{1}\right\}  ,B^{3}=\left\{  P_{2}\right\}  \text{.}%
\]
This game has some interesting properties; they will be fully described in a
separate publication, but as an example suppose it is played on the star graph
of Figure \ref{fig0701a}, with initial positions as indicated. It is easily
checked that, even though the star graph is cop-win, the game has \emph{only
}noncapturing NE. 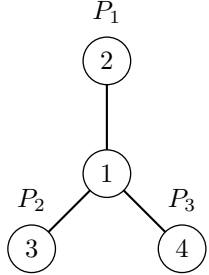
\begin{figure}[ptbh]
\begin{center}
\begin{tikzpicture}
\SetGraphUnit{2}
\Vertex[x= 0,y= 0]{1}
\Vertex[x= 0,y= 1.5]{2}
\Vertex[x=-1,y=-1]{3}
\Vertex[x= 1,y=-1]{4}
\node(A) [label=$P_1$] at (0,1.75) {};
\node(B) [label=$P_2$] at (-1,-0.75) {};
\node(C) [label=$P_3$] at (1,-0.75) {};
\Edge(1)(2)
\Edge(1)(3)
\Edge(1)(4)
\SetVertexNoLabel
\end{tikzpicture}
\end{center}
\par
\label{fig0701}\caption{On this graph Cyclic Cops and Robbers have only
noncapturing NE.}%
\label{fig0701a}%
\end{figure}

Many similar games can be constructed along these lines and all of them will
(i) fall within the game-theoretic framework of Section \ref{sec02} and
(ii)\ possess a well defined game theoretic solution, namely a NE\ in
deterministic positional strategies, according to Theorem \ref{prop0301}.

In fact the framework of Section \ref{sec02} can accommodate further
generalizations for which Theorem \ref{prop0301} will still hold. We list some
additional generalizations to the idea of graph pursuit game.

\begin{enumerate}
\item \emph{Payoffs}. The turn payoffs $q^{n}$ can take values in $\left[
-1,1\right]  $ rather than $\left\{  -1,1\right\}  $. As an example, we have
introduced and studied the game of \emph{Selfish Cops and Robbers}
\cite{kehagias2017}, in which two cops pursue a robber but do not split the
capture payoff equally; instead the capturing (resp. noncapturing)\ cop
receives payoff $\left(  1-\varepsilon\right)  $ (resp. $\varepsilon$), where
$\varepsilon\in\left[  0,\frac{1}{2}\right]  $. Hence each cop has a motive to
be the one who actually captures the robber; if this \textquotedblleft
selfishness\textquotedblright\ is sufficiently strong (this will depend on the
$\varepsilon$ value) it can be exploited by the robber to avoid capture ad infinitum.

\item \emph{Teams}. So far we have assumed that each player controls a single
token. But we can also assume that a game is played by $N$ players (with
$N\geq2$) with $P_{n}$ controlling $K_{n}$ tokens. An example of this is the
classic CR\ game with more than one cop tokens (all of them controlled by a
single \emph{cop player}). Another example are the $\widetilde{\Gamma}_{N}%
^{n}\left(  G|s_{0}\right)  $ auxiliary games of Sections \ref{sec0501} and
\ref{sec06}. These are two-player games, but the idea can be applied to
multi-player games as well. \ For example we could have the three-player
GCR\ game with $P_{1}$ controlling two pursuer tokens and each of $P_{2}$ and
$P_{3}$ controlling one pursuer and one evader token.

\item \emph{Game termination}. So far we have assumed that the game terminates
upon the first capture, but this can also be modified. For example the game
could end upon the elimination of all tokens of one player, or when no more
captures are possible.
\end{enumerate}

Since all of the above modifications can be accommodated by the formulation of
Section \ref{sec02}, the respective games can be analyzed by game-theoretic
methods. At the very least, by Theorem \ref{prop0301} they all possess NE;
further results can be obtained by exploiting the special characteristics of
each game.

\section{Conclusion\label{sec08}}

In this paper we have introduced and studied the \emph{Generalized Cops and
Robbers} game $\Gamma_{N}\left(  G|s_{0}\right)  $, a multi-player pursuit
game in graphs. The two-player version $\Gamma_{2}\left(  G|s_{0}\right)  $ is
essentially equivalent to the classic CR\ game. The three-player version
$\Gamma_{3}\left(  G|s_{0}\right)  $ can be understood as two CR games played
simultaneously on the same graph; a player can simultaneously be pursuer
\emph{and }evader. This also holds for $\Gamma_{N}\left(  G|s_{0}\right)  $
when $N\geq4$.

Using a formulation of $\Gamma_{N}\left(  G|s_{0}\right)  $ as a discounted
stochastic game of perfect information we have proved that it has at least one
NE in positional deterministic strategies. Using auxiliary two-player games
$\widetilde{\Gamma}_{N}^{n}\left(  G|s_{0}\right)  $ we have also proved the
existence of an additional NE\ in nonpositional deterministic strategies. We
have also studied the capturing properties of the $\Gamma_{N}\left(
G|s_{0}\right)  $ NE\ in connection to the cop-number $c\left(  G\right)  $.

Both $\Gamma_{N}\left(  G|s_{0}\right)  $ and $\widetilde{\Gamma}_{N}%
^{n}\left(  G|s_{0}\right)  $ are members of a general family of graph pursuit
games, which can be described by the framework of Section \ref{sec02} and its
generalizations, presented in Section \ref{sec07}. This family is a broad
generalization of the two-player graph pursuit games previously studied to the
multi-player case; it contains games with rather unexpected properties and
hence, we believe, it deserves additional study.

\end{document}